\documentclass{l4dc2025}

\usepackage{amsfonts, bbm, empheq, enumitem, float, mathrsfs}

\numberwithin{equation}{section} 
\numberwithin{theorem}{section}

\let\set\relax %
\DeclarePairedDelimiter\set\{\}
\newcommand\E{\mathbb{E}} 
\newcommand\R{\mathbb{R}}

\newcommand\1{\mathbbm{1}}
\newcommand\N{\mathbbm{N}}
\renewcommand*{\P}{\mathbb{P}}
\def\ddefloop#1{\ifx\ddefloop#1\else\ddef{#1}\expandafter\ddefloop\fi}

\def\ddef#1{\expandafter\def\csname #1hat\endcsname{\ensuremath{\widehat{\csname #1\endcsname}}}}
\ddefloop {alpha}{beta}{gamma}{delta}{epsilon}{varepsilon}{zeta}{eta}{theta}{vartheta}{iota}{kappa}{lambda}{mu}{nu}{xi}{pi}{varpi}{rho}{varrho}{sigma}{varsigma}{tau}{upsilon}{phi}{varphi}{chi}{psi}{omega}{Gamma}{Delta}{Theta}{Lambda}{Xi}{Pi}{Sigma}{varSigma}{Upsilon}{Phi}{Psi}{Omega}{ell}\ddefloop

\def\ddef#1{\expandafter\def\csname #1scr\endcsname{\ensuremath{\mathcal{#1}}}}
\ddefloop ABCDEFGHIJKLMNOPQRSTUVWXYZ\ddefloop

\def\ddef#1{\expandafter\def\csname #1cal\endcsname{\ensuremath{\mathscr{#1}}}}
\ddefloop ABCDEFGHIJKLMNOPQRSTUVWXYZ\ddefloop

\def\ddef#1{\expandafter\def\csname #1hat\endcsname{\ensuremath{\widehat{#1}}}}
\ddefloop ABCDEFGHIJKLMNOPQRSTUVWXYZ\ddefloop

\def\ddef#1{\expandafter\def\csname #1hat\endcsname{\ensuremath{\widehat{#1}}}}
\ddefloop abcdefghijklmnopqrsuvwxyz\ddefloop

\def\ddef#1{\expandafter\def\csname #1bar\endcsname{\ensuremath{\bar{#1}}}}
\ddefloop ABCDEFGHIJKLMNOPQRSTUVWXYZ\ddefloop

\def\ddef#1{\expandafter\def\csname #1bar\endcsname{\ensuremath{\overline{#1}}}}
\ddefloop ABCDEFGHIJKLMNOPQRSTUVWXYZ\ddefloop

\def\ddef#1{\expandafter\def\csname #1bar\endcsname{\ensuremath{\bar{#1}}}}
\ddefloop abcdefghijklmnopqrstuvwxyz\ddefloop

\def\ddef#1{\expandafter\def\csname #1bar\endcsname{\ensuremath{\bar{\csname #1\endcsname}}}}
\ddefloop {alpha}{beta}{gamma}{delta}{epsilon}{varepsilon}{zeta}{eta}{theta}{vartheta}{iota}{kappa}{lambda}{mu}{nu}{xi}{pi}{varpi}{rho}{varrho}{sigma}{varsigma}{tau}{upsilon}{phi}{varphi}{chi}{psi}{omega}{Gamma}{Delta}{Theta}{Lambda}{Xi}{Pi}{Sigma}{varSigma}{Upsilon}{Phi}{Psi}{Omega}{ell}\ddefloop

\def\ddef#1{\expandafter\def\csname #1bb\endcsname{\ensuremath{\mathbb{#1}}}}
\ddefloop ABCDEFGHIJKLMNOPQRSTUVWXYZ\ddefloop

\def\ddef#1{\expandafter\def\csname #1tilde\endcsname{\ensuremath{\widetilde{#1}}}}
\ddefloop ABCDEFGHIJKLMNOPQRSTUVWXYZ\ddefloop

\def\ddef#1{\expandafter\def\csname #1tilde\endcsname{\ensuremath{\tilde{#1}}}}
\ddefloop abcdefghijklmnopqrstuvwxyz\ddefloop

\def\ddef#1{\expandafter\def\csname #1tilde\endcsname{\ensuremath{\tilde{\csname #1\endcsname}}}}
\ddefloop {alpha}{beta}{gamma}{delta}{epsilon}{varepsilon}{zeta}{eta}{theta}{vartheta}{iota}{kappa}{lambda}{mu}{nu}{xi}{pi}{varpi}{rho}{varrho}{sigma}{varsigma}{tau}{upsilon}{phi}{varphi}{chi}{psi}{omega}{Gamma}{Delta}{Theta}{Lambda}{Xi}{Pi}{Sigma}{varSigma}{Upsilon}{Phi}{Psi}{Omega}{ell}\ddefloop

\newcommand{\xpln}[1]{, \ \text{#1}} %

\DeclareMathOperator*{\amax}{arg\,max}

\DeclarePairedDelimiter\absBase{\lvert}{\rvert}
\newcommand\abs[1]{\absBase*{#1}}

\DeclarePairedDelimiterXPP{\nucnormBase}[1]{}\lVert\rVert{_\mathrm{nuc}}{#1} %

\DeclarePairedDelimiterXPP{\opnormBase}[1]{}\lVert\rVert{_\mathrm{2}}{#1} %
\newcommand\opnorm[1]{\opnormBase*{#1}}
\DeclarePairedDelimiterXPP{\FnormBase}[1]{}\lVert\rVert{_\mathrm{F}}{#1} %
\newcommand\Fnorm[1]{\FnormBase*{#1}}
\DeclarePairedDelimiterXPP{\enormBase}[1]{}\lVert\rVert{_2}{#1} %
\newcommand\enorm[1]{\enormBase*{#1}}

\newcommand\matb[1]{\begin{bmatrix} #1 \end{bmatrix}}

\DeclareMathOperator{\Unif}{Unif}

\DeclareMathOperator{\Tr}{Tr}

\newcommand{\pseudo}[1]{ {#1}^\dagger}

\DeclarePairedDelimiter\lrpBase{(}{)}
\newcommand\lrp[1]{\lrpBase*{#1}}
\DeclarePairedDelimiter\lrbBase{[}{]}
\newcommand\lrb[1]{\lrbBase*{#1}}
\DeclarePairedDelimiter\brkBase{\langle}{\rangle} %
\newcommand\brk[1]{\brkBase*{#1}}
\DeclarePairedDelimiterXPP\ExpBase[1]{\E}{[}{]}{}{#1} %
\newcommand\Exp[1]{\ExpBase*{#1}}

\DeclarePairedDelimiter\floor{\lfloor}{\rfloor}

\let\Pr\relax
\DeclarePairedDelimiterXPP\PrBase[1]{\P}{[}{]}{}{#1}
\newcommand\Pr[1]{\PrBase*{#1}}
\newcommand{\lsim}{\lesssim}
\newcommand{\gsim}{\gtrsim}
\newcommand{\eps}{\varepsilon}
\newcommand{\by}{\times}

\newcommand{\iid}{\overset{iid}{\sim}}

\newcommand\ot{\otimes}
\DeclareMathOperator*{\poly}{poly}
\renewcommand\u[1]{^{(#1)}}
\newcommand\notex[1]{}
\newcommand\pmin{p_{\min}}

\DeclarePairedDelimiterXPP{\subGnorm}[1]{}\lVert\rVert{_{\psi_2}}{#1} %
\DeclarePairedDelimiterXPP{\indBase}[1]{\1}{(}{)}{}{#1}
\newcommand\ind[1]{\indBase*{#1}}

\newcommand\Rstart{{R_{\mathrm{start}}}}
\newcommand\Riter{{R_{\mathrm{iter}}}}

\newcommand\clean{M^W_3}
\newcommand\first{\Mhat^W_3}
\newcommand\firstV{\Mhat^V_3}

\newcommand\mixed{M^{\What}_3}
\newcommand\tiltil{\Mtilde^{\Wtilde}_3}
\newcommand\mixedtil{M^{\Wtilde}_3}

\DeclareMathOperator{\subG}{subG}

\newcommand\bigG{\mathscr{G}}

\newcommand\uL{^{(L)}}
\newcommand\gscrL{\gscr^{(L)}}
\newcommand\gL{g\uL}

\usepackage{autonum} %

\newcommand{\td}[1]{}
\newcommand{\note}[1]{}
\newcommand{\tdx}[1]{}

\usepackage[scr=esstix]{mathalpha} %
\newcommand\gscr{\mathscr{g}}

\title[Learning Mixtures of Linear Systems]{Finite Sample Analysis of Tensor Decomposition for Learning Mixtures of Linear Systems}
\usepackage{times}

\coltauthor{\Name{Maryann Rui} \Email{mrui@mit.edu}\and
\Name{Munther A. Dahleh} \Email{dahleh@mit.edu}\\
\addr{Laboratory for Information and Decision Systems, Massachusetts Institute of Technology, Cambridge, MA 02139, USA}
}

\begin{document}
\maketitle
\begin{abstract}
We study the problem of learning mixtures of linear dynamical systems (MLDS) from input-output data. The mixture setting allows us to leverage observations from related dynamical systems to improve the estimation of individual models. Building on spectral methods for mixtures of linear regressions, we propose a moment-based estimator that uses tensor decomposition to estimate the impulse response parameters of the mixture models. The estimator improves upon existing tensor decomposition approaches for MLDS by utilizing the entire length of the observed trajectories. We provide sample complexity bounds for estimating MLDS in the presence of noise, in terms of both the number of trajectories $N$ and the trajectory length $T$, and demonstrate the performance of the estimator through simulations.
\end{abstract}


\section{Introduction}
In many domains of learning time series, such as in healthcare, social sciences, and biological sciences \citep{ernst2005clustering}, there are often a large number of data sources (e.g., patients, systems, cells), but a limited amount of data from each individual source. Without additional assumptions, it may be impossible to identify individual models for each data source. However, when the data is actually generated from a few underlying models, we can leverage the collective observations to learn these models, which can then be used to improve estimates of individual systems. The setting of mixture models, in particular, allows for tractability in learning multiple models from data. 

In this paper, we propose and study a moment-based estimator that uses tensor decomposition to learn mixtures of linear dynamical systems (MLDS) from input-output data. 
Compared to existing methods, which we detail in the following related work section, the estimator allows us to utilize the full length $T$ of the observed trajectories. We also provide explicit sample complexity bounds for estimating stable MLDS in the presence of process and observation noise, showing how the error depends on system parameters and how increasing both the number of trajectories, $N$, and the individual trajectory length, $T$, can be leveraged to improve estimation. 

In the remainder of this section, we review related work. In Section \ref{sec:setup}, we formalize our MLDS model. In Section \ref{sec:method}, we introduce the moment-based MLDS estimator and an estimator for mixtures of linear regression (MLR), on which the MLDS estimator is built. In Section \ref{sec:analysis}, we provide sample complexity bounds for the estimator in Proposition \ref{prop:MLDS-result}, where as a key step we derive finite sample error bounds for learning mixtures of linear regressions in the presence of independent noise and bounded perturbations. Finally, in Section \ref{sec:simulations}, we demonstrate the performance of the tensor decomposition approach to MLDS through simulations. The Appendix contains proofs and auxiliary results. 

\subsection{Related Work}\label{sec:related-work}
Our work lies at the intersection of spectral methods for mixtures of linear regression and system identification for partially observed linear systems. The most relevant work is \cite{bakshi2023tensor}, which also sits at this intersection, and which inspired us to derive an alternate moment estimator with explicit sample complexity guarantees for the MLDS problem.
While recent work has also studied other forms of shared structure between multiple linear dynamical systems, such as a shared low dimensional representation of the transition matrix \citep{modi2021joint, zhang2023multi}, these are largely restricted to fully observed systems. Thus, we choose to focus our review of related work on \textit{mixture} models of static and dynamical linear systems.

\vspace{-0.3em}
\paragraph{Mixtures of linear regression.}
In mixtures of linear regressions, data of the form $\set{(x_i, y_i)}_{i \in [N]}$ is observed, with a generating model given by $y_i = \brk{x_i, \beta_i}$, where the parameter $\beta_i$ is sampled from a given distribution over the $K$ mixture components $\set{\beta_k}_{k \in [K]}$. 
The goal is to learn the $K$ mixture components and their respective mixture weights. Approaches to solving MLR can generally be grouped into those based on tensor decomposition \citep{anandkumar2014tensor}, alternating minimization \citep{yi2014alternating}, and gradient methods \citep{li2018learning}, or a combination of these. Both \cite{zhong2016mixed} and \cite{yi2016solving} %
apply tensor decomposition on sixth-order moments to initialize iterative algorithms based on gradient descent and alternating minimization, respectively. 
While they provide sample complexity guarantees for MLR in the \textit{noiseless} setting, we extend their estimator and analyses to the setting of noisy observations of linear system trajectories. 

\vspace{-0.3em}
\paragraph{Mixtures of linear dynamical systems.} 
\cite{chen2022learning} also study learning mixtures of dynamical systems, though restricted to the fully-observed setting. Most relevant to our work is \cite{bakshi2023tensor}, which introduces a moment-based estimator that uses tensor decomposition to prove that under minimal assumptions, mixtures of linear dynamical systems (MLDS) can be learned with polynomial sample and computational complexity. 
However, they do not provide explicit sample complexity bounds and their algorithm only uses a fixed number of samples from each observed trajectory, forfeiting possibly useful information in longer trajectories. In this work, we provide a different moment-based estimator that utilizes the entire length of observed trajectories and derive explicit finite-sample error bounds for mixtures of stable systems with a sharper $\poly(\ln(1/\delta))$ dependence (versus $\poly(1/\delta)$ in \cite{bakshi2023tensor}), where bounds are given with high probability $1-\delta$, and include the effects of process and measurement noise. A detailed comparison of the estimators is given in Section \ref{sec:bakshi-comparison}.

\vspace{-0.3em}
\paragraph{Finite sample bounds for linear system identification.} 
There is a large body of work on the identification of partially observed linear systems from input-output data, with recent works providing finite-sample error bounds for learning from a single trajectory, or rollout. 
A standard approach is to estimate Markov parameters of the system and then use the Ho-Kalman, or eigensystem realization, algorithm \citep{ho1966effective} to obtain a state space realization of the system. \cite{sarkar2021finite} and \cite{oymak2021revisiting} estimate Markov parameters from a single trajectory of strictly stable systems using an ordinary least squares estimator. \cite{bakshi2023new} derive a moment-based estimator for the Markov parameters, though the estimator coefficients must be computed via a separate convex program. %
Estimating single partially observed systems from \textit{multiple} rollouts has also recently been studied. \cite{zheng2020non} provide error bounds for an OLS estimator on $N$ independent length $T$ trajectories, for both stable and unstable systems. However, the error in estimating the first $T$ Markov parameters grows superlinearly in the trajectory length $T$, which is a suboptimal trend for strictly stable systems. 

\section{Setup}\label{sec:setup}
\subsection{Notation} 
For any natural number $N \in \N$, we define the set $[N] \coloneq \set{1,2, \dots,  N}$. For a $d_1 \by d_2$ matrix $A$, we denote its trace $\Tr(A)$, transpose $A'$, Moore-Penrose pseudoinverse $\pseudo{A}$, Frobenius norm $\Fnorm{A}$, and operator (spectral) norm $\opnorm{A}$. For $i \in [\min(d_1, d_2)]$, $\sigma_i(A)$ is the $i$-th largest singular value of $A$.
The identity matrix in $\R^{d\by d}$ is denoted $I_d$. 
A real-valued random variable $X$ is subgaussian with variance proxy $\sigma_x^2$ if $\Pr{\abs{X} \geq t} \leq 2\exp(-t^2/(2\sigma_x^2))$ for $t>0$. If in addition, $X$ is zero mean, we write $X \sim \subG(0, \sigma^2)$. Similarly, $X$ is subexponential with parameter $\lambda$ if $\Pr{\abs{X} \geq t} \leq 2 \exp(-t/\lambda)$. A random vector $X$ is subgaussian if for all fixed vectors $v \in \R^n$, $\brk{X,v}$ is subgaussian \citep{vershynin2018high}.
We use $c$ to denote a universal positive constant, which may vary from line to line. For real-valued functions $a, b$, the inequality $a \lsim b$ implies $a \leq cb$ for some $c$. 
Unless otherwise specified, all random variables are defined on the same probability space. 
\paragraph{Tensors.} 
A $K$-th order tensor in a Euclidean space is an element of the tensor product of $K$ Euclidean spaces. The tensor product, or outer product, of $K$ vectors $\set{v_k \in \R^{d_k}}_{k \in [K]}$ is denoted $v_1 \ot v_2 \ot \cdots \ot v_K$ and is a rank 1 $K$-th order tensor with $(i_1, i_2, \dots, i_K)$-th entry equal to $\prod_{k=1}^K v_k(i_k)$. 
For a vector $v \in \R^d$, $v^{\ot K} = v \ot v \ot \cdots \ot v$ ($K$ times) is its $K$-th tensor power.
In general, the rank of a tensor $M$ is the smallest number of rank-one tensors such that $M$ can be expressed as their sum. 
A third-order $d_1 \by d_2 \by d_3$ tensor $M$ of rank $r$ may thus be written as $M=\sum_{i=1}^r a_i \ot b_i \ot c_i$ for some $a_i \in \R^{d_1}, b_i \in \R^{d_2}, c_i \in \R^{d_3}$. Viewing such a tensor $M$ as a multilinear map, we have the mapping for matrices $A\in \R^{d_1 \by l_1}, B\in \R^{d_2 \by l_2}$ and $C\in \R^{d_3 \by l_3}$, $M(A, B, C) = \sum_{i=1}^r A'a_i \ot B'b_i \ot C' c_i \in \R^{l_1} \ot \R^{l_2} \ot \R^{l_3}$. 

A symmetric third-order tensor $M$ is invariant under permutations of its arguments $(A, B, C)$. 
Its operator norm is defined as $\opnorm{M} = \sup_{a\in \Sscr^{d-1}} \abs{M(a,a,a)}$. For simplicity of notation, given a $d\by K$ matrix $W$, we sometimes use the shorthand $M^W\coloneq M(W, W, W)$.
See \cite{kolda2009tensor} for an introductory reference on tensors and tensor decomposition.

\subsection{Model}\label{sec:model}
\paragraph{Mixture model.} 
A partially-observed, strictly causal, linear-time invariant (LTI) system can be represented in terms of its impulse response $g = (g(1),\ g(2),\, \dots)$, which captures the input-output mapping of the system. Assuming for simplicity $m$-dimensional inputs and single-dimensional outputs, the $j$th impulse response, or Markov, parameter, $g(j)$ is an $m\by 1$ vector, for $j \in \N$. Given an input trajectory $\set{u_t\in \R^m}_{ t \in \N}$, the output of the system at each time $t\in \N$ is given by
\begin{align}\label{eq:new-linear-output}
y_t = \sum_{j=1}^t \brk{g(j),u_{t-j} + w\u{1}_{t-j}} + w\u{2}_t,%
\end{align}
where $w\u{1}_t\in \R^m$ and $w\u{2}_t \in \R$ represent process and measurement noise, respectively, at time $t$. We assume zero inputs $u_t = 0$ for $t<0$ and zero feedthrough (i.e., $y_t$ does not depend on $u_t$).

Consider a mixture of $K\geq 2$ LTI models given by $\bigG = \set{(\gscr_k, p_k)}_{k \in [K]}$, where the model with impulse response sequence $\gscr_k$ has associated probability $p_k>0$, with $\sum_{k=1}^K p_k = 1$. 
We observe $N$ input-output trajectories of length $T$ in the data set $\Dscr = \set{(u_{i,t-1}, y_{i,t}) \mid i \in [N], t \in [T]}$, which are generated from the mixture model in the following way: For each trajectory $i \in [N]$, a system model $g_i = \gscr_{k_i}$ is drawn from $\bigG$, where the index $k_i=k$ is drawn with probability $p_{k}$, for $k \in [K]$. A trajectory is then rolled out with randomly generated inputs $\set{u_{i,t-1}}_{t \in [T]}$ and corresponding outputs $\set{y_{i,t}}_{t \in [T]}$ generated according to \eqref{eq:new-linear-output}. 

\begin{remark}
Partially-observed LTI systems corresponding to \eqref{eq:new-linear-output} are often represented by the following input-state-output dynamics with a state variable $x_t\in \R^n$, where $n$ is the minimal order of the system: 
\begin{align}\label{eq:st-sp-model}
x_{t+1} &= A x_{t} + B (u_{t} + w^{(1)}_{t}), \quad 
y_{t} = C x_{t} + w^{(2)}_{t}.
\end{align}
$A \in \R^{n \by n}$ is the state transition matrix, $B \in \R^{n \by m }$ the control matrix, and $C \in \R^{1\by n}$ the measurement matrix. %
While the parameters $(C, A, B)$ representing the system are only identifiable up to a similarity transformation, they correspond to the representation-independent Markov parameters by %
$g(t) = CA^{t-1}B$ for $t\geq 1$. %
Because the crux of most time-domain system identification methods, including ours, lies in estimating Markov parameters, we focus on the impulse-response representation and provide pointers to state-space estimation when relevant.
\end{remark}
\paragraph{Objective.} Given an input-output data set $\Dscr$ of length-$T$ trajectories from $N$ systems, we aim to estimate the generating mixture $\bigG$ comprising the component models $\gscr_k$ and their weights $p_k$. To do so, it suffices to learn just a finite number of Markov parameters to identify the infinite impulse response sequence. If an LTI system given by $\gscr_k$ has finite order bounded by $n>0$, the sequence $\gscr_k$ is completely determined by its first $2n+1$ elements \citep{gragg1983partial}. Thus, it suffices to learn the first $L\geq 2n+1$ Markov parameters of each of the $K$ models in mixture. Further, even if $L < 2n+1$, the first $L$ Markov parameters can still be very informative of the system behavior.
To this end, for a fixed $L$ such that $1 \leq L \leq T$, let us define the truncated impulse response vector $\gL_i = \lrb{g_i(1)',\ \dots ,\ g_i(L)' }' \in \R^{Lm}$ for the system generating the $i$th trajectory, $i \in [N]$. We focus on learning the first $L$ Markov parameters and weights $\set{(\gscrL_k, p_k)}_{k \in [K]}$ of the mixture. %

\subsection{Assumptions}\label{sec:Assumptions}
\paragraph{Dynamics and distributional assumptions.} 
We assume that each of the $K$ models in the mixture are strictly stable. 
Under this assumption, define the finite quantity $\Gamma(\gscr_k) \coloneq 1+\sum_{t=1}^{\infty} \enorm{g_k(t)}^2$ capturing the energy of each system $k \in [K]$, %
and $\Gamma_{\max}\coloneq  \max_{k \in [K]} \Gamma(\gscr_k) < \infty$. 
Furthermore, let $\rho>0$ and $C_{\rho} >0$ be such that for every $t \in \N$, $\max_{k \in [K]} \enorm{g_k(t)} \leq C_\rho \rho^t$. For example, we can take any $\rho<1$ greater than the largest spectral radius of the $K$ models, by Gelfand's formula \citep{kozyakin2009accuracy}. 

For each trajectory $i \in [N]$, for $t\geq 0$, we assume that the inputs $u_{i,t}$ are i.i.d. zero-mean isotropic Gaussian random vectors in $\R^m$ with variance $\sigma_u^2 I_m$, and that the process noise $w^{(1)}_{t}\in \R^m$ and measurement noise $w^{(2)}_{t}\in \R$ are independent zero-mean subgaussian random vectors with variance proxies $\sigma_{w^{(1)}}^2$, and $\sigma_{w^{(2)}}^2$, respectively. Let $\sigma_w \coloneq \max(\sigma_{w^{(1)}}, \sigma_{w^{(2)}})$.

\paragraph{Mixture assumptions.} Let $\pmin := \min_{k \in [K]} p_k >0$ be a lower bound on the mixture weights. 
We also assume a non-degeneracy condition on the vectors of Markov parameters $\gscrL_k$. %
Let $M_2 \coloneq \sum_{k=1}^K p_k \gscrL_k \ot \gscrL_k$ be the weighted sum of outer products of the mixture components. Then we assume that $\sigma_K(M_2) > 0$. With abuse of notation (as it will be clear from context), we let $\sigma_K$ denote $\sigma_K(M_2)$. 
Note that we do not assume a minimum separating distance between pairs of mixture components, %
but that the non-degeneracy condition does require the number of components $K \leq d$, which is a reasonable setting in many practical applications. 

\section{Method}\label{sec:method}
Recall that we aim to estimate the parameters $\set{(\gscrL_k, p_k)}_{k \in [K]}$ of the mixture model $\bigG$, which are sufficient to identify $\bigG$ and to yield minimal state-space realizations of the models when $L \geq 2n+1$. 
We first identify our problem of estimating Markov parameters $\gL_i = \gscr_{k_i}\uL$ in MLDS with elements of a linear regression model, but with additional noise and correlated perturbations. 
To do so, we express $y_{it}$ in the form of a linear regression with coefficients $\gL_i$. Define the parameter vector $f\uL_{i}:= \big[ 1, g^{(L)\prime}_i\big]' \in \R^{1+Ln}$, the vector of concatenated inputs from $t-L$ to $t-1$, $\ubar_{i,t}:=\big[u_{i,t-1}',\ u_{i,t-2}',\ \dots,\ u_{i,t-L}']' \in \R^{Lm},$ and the vector of concatenated noise variables $\wbar_{i,t} = \big[w^{(2)}_{i,t},\ (w^{(1)}_{i,t-1})',\ \cdots,\ (w^{(1)}_{i,t-L})'\big]' \in \R^{1+Ln}$. Then the output $y_{it}$ can be written as
\begin{align}\label{eq:linear-output-2}
y_{it} = \brk{\gL_i, \ubar_{it}} + \brk{f\uL_i, \wbar_{it}} + \xi_{it},
\end{align}
where we collect the remainder due to the length $L$ truncation of the impulse response in the term
\begin{align} \label{eq:xi}
\xi_{it} = \sum_{j=L+1}^t \brk{g_i(j), u_{i,t-j} + w_{i,t-j}^{(1)}}.
\end{align}

Since the covariates $\set{\ubar_{i,t}}$ in \eqref{eq:linear-output-2} are vectors of lagged inputs, covariates that are close in time (e.g., $\ubar_{i,t}$ and $\ubar_{i,t+1}$) have overlapping entries and are thus dependent.
In order to work with independent covariates across observations, which simplifies the later analysis, we simply take every $L$-th sample starting at time index $L$. Assume without loss of generality that $L$ divides $T$ (otherwise discard at most $L-1$ samples at the end of the trajectory), and let $\Jscr = \set{L,2L, 3L, ..., T}$ be an index set of size $T/L$. Then the vectors of lagged inputs and noise terms indexed by $\Jscr$, $\set{\ubar_{i,t},\wbar_{it} \mid t\in \Jscr, i \in [N]}$, are mutually independent random vectors. %
We can thus view the MLDS data set $\set{(\ubar_{it}, y_{it}) \mid t \in \Jscr, i \in [N]}$ as being sampled from a mixture of linear regressions with noise and perturbations as formulated in Definition \ref{def:MLR}, with an effective sample size of $NT/L$, mapping each index $(i,t) \in [N] \by \Jscr$ to a linear index $j \in [NT/L]$.

\IncMargin{1em}%
\begin{algorithm2e}[bthp] %
\SetAlgoLined  %
\SetKwFor{For}{for}{:}{} 
\LinesNumbered
\DontPrintSemicolon

\KwIn{%
\begin{minipage}[t]{\linewidth}
    $\set{(u_{i,t-1}, y_{i,t}) \mid t \in [T], i \in [N]}$ --- Input-output trajectories \\
    $\Nscr_2 \cup \Nscr_3 = [N]$ --- Index set partition  for estimating moments \\
    $L$ --- Number of Markov parameters to estimate ($L\leq T$)\\
    $K$ --- Number of mixture components
  \end{minipage}
}
\KwOut{$\set{(\phat_k, \widehat{\gscr}_k\uL)}_{\mid k \in [K]}$ --- Markov parameters and weights of mixture
}

$\Jscr \gets \set{L, 2L, ..., \floor{T/L}}$ \;
\For{$i \in [N], t \in \Jscr$}{
	$\ubar_{i,t} \gets \matb{u_{i,t-1}'& \cdots & u_{i,t-L}'}'$ \tcp*{Stack inputs}
}
$\set{(\phat_k, \widehat{\gscr}_k\uL)}_{k \in [K]} \gets $Algorithm \ref{alg:MLR-estimator}$(\set{(\ubar_{i,t}, y_{i,t}) \mid i \in [N], t \in \Jscr}, \Nscr_2\by \Jscr, \Nscr_3\by \Jscr)$ \tcp*{Mixture of Linear Regressions}
\caption{Mixtures of Linear Dynamical Systems Estimator}\label{alg:MLDS-estimator}
\end{algorithm2e}
\DecMargin{1em} %

\begin{definition}[Mixture of Linear Regression]\label{def:MLR}
Data $\set{(x_i, \ytilde_i) \in \R^d \by \R}_{i \in [N]}$ is generated by a mixture of linear regressions with noise and perturbations if the output $\ytilde_i$ can be expressed as 
\begin{align}
\ytilde_i = \brk{x_i, \beta_{k_i}} + \eta_i + \xi_i,
\end{align}
where the covariates $x_i \iid \Nscr(0, I_d)$ are independent zero-mean isotropic gaussians, the term $\eta_i\sim \subG(0, \sigma_{\eta}^2)$ represents independent subgaussian noise, and the term $\xi_i \sim \subG(0, \sigma_{\xi}^2)$ represents an additional subgaussian perturbation which may be correlated with the covariates and noise terms $\set{x_j, \eta_j}_{i \in [N]}$. 
The latent variable $k_i$ indicates the mixture component with coefficient $\beta_{k_i} \in \set{\beta_k}_{k \in [K]}$ that the $i$-th observation belongs to, where $\Pr{k_i=k} = p_k$ for $k \in [K]$. 
\end{definition}

Note that the noise term $\big\langle f_i\uL, \wbar_{it}\big \rangle$ in \eqref{eq:linear-output-2} is zero-mean subgaussian with variance proxy $\sigma_{w}^2 \big\lVert f_i\uL\big \rVert^2 \leq \sigma_w^2 \Gamma_{\max}$.
The perturbations $\xi_{it}$ are not necessarily independent of the covariates or noise, but they are zero-mean subgaussian and thus can be bounded with high probability. Indeed, when the linear models in the mixture are strictly stable, the effect of past inputs on the present output decreases exponentially in $L$, the rate of which is captured by $\rho$. Thus, if $L$ is large enough we can treat the contributions of the remaining Markov parameters and past inputs in $\xi_{i, t}$ as bounded noise.

We complete the mapping of the MLDS problem to the MLR problem in Definition \ref{def:MLR} by assigning the covariates $x_j \gets  \ubar_{i,t}/\sigma_u = \matb{u_{i,t-1}' & \cdots & u_{i,t-L}'}'/\sigma_u \in \R^{Lm}$, outputs $\ytilde_j \gets y_{i,t}$ coefficients $\beta_j \gets \sigma_u \gL_i \in \R^{Lm}$, independent zero-mean subgaussian noise $\eta_j \gets \big\langle f_i\uL, \wbar_{it}\big\rangle$, and subgaussian perturbations $\xi_j \gets \xi_{it}$ as defined in \eqref{eq:xi}, again with the mapping of indices $(i,t) \mapsto j \in [NT/L]$. Algorithm \ref{alg:MLDS-estimator} constructs the MLR problem in this way, and then uses Algorithm \ref{alg:MLR-estimator} as the key subroutine to obtain Markov parameter estimates of the mixture components. From there, the Ho-Kalman algorithm can be used to obtain state-space realizations for the mixture.  

\subsection{Mixtures of Linear Regression with Noise and Perturbations}
In this section we detail the tensor decomposition approach for solving MLR (Definition \ref{def:MLR}), which is the workhorse of Algorithm \ref{alg:MLDS-estimator} for solving MLDS. 

\paragraph{Motivation for tensor decomposition.} 
While a matrix, or second-order tensor, $M_2$ of rank $K$ can be expressed as a sum of $K$ rank-1 matrices, e.g., $M_2 = \sum_{k=1}^K a_k \ot b_k$, this decomposition is not unique. On the other hand, under mild assumptions, a third-order tensor $M_3$ of rank $K$ does have a unique decomposition as a sum of $K$ rank-1 tensors (up to scaling and ordering of factors). 
In the case of a symmetric tensor $M_3 = \sum_{k=1}^K p_k \beta_k^{\ot 3}$, a sufficient condition for uniqueness of the decomposition is when $\set{\beta_k}_{k \in [K]}$ are linearly independent. Then the set of summands $p_k \beta_k^{\ot 3}$ is unique, though the scaling between $p_k$ and $\beta_k$ needs to be resolved separately. If  $\set{(p_k, \beta_k)}_{k \in [K]}$ represent the parameters of a mixture, then knowing $M_3$ would allow us to recover the mixture model through tensor decomposition. Additionally, if we have a noisy estimate of $M_3$, results on the robustness of tensor decomposition for non-degenerate tensors \citep{anandkumar2014tensor} assure us that the estimated components are not too far from the true components.

\paragraph{Estimating MLR.} We now extend the moment-based tensor decomposition approach to estimating mixtures of linear regressions that was presented in \citep{yi2016solving,zhong2016mixed} and given in Algorithm \ref{alg:MLR-estimator}. While these works provide estimation error bounds in the \textit{noiseless} case, in Section \ref{sec:analysis} we provide performance guarantees under both i.i.d. noise $\eta_i$ and bounded perturbations $\xi_i$, which may be correlated with other variables. To begin our analysis of the MLR algorithm, we examine the moments estimated by Algorithm \ref{alg:MLR-estimator} on the noisy, perturbed linear regression data: 
\begin{align}\label{eq:M23tilde-MLR}
\Mtilde_2 &= \frac{1}{2N_2} \sum_{i\in\Nscr_2} \ytilde_i^2 \lrp{x_i\ot x_i - I_d},\ \text{ and }\
\Mtilde_3 = \frac{1}{6N_3} \sum_{i\in \Nscr_3} \ytilde_i^3 \lrp{x_i^{\ot 3} - \Escr(x_i)},
\end{align}
where $\Escr(x_i) = \sum_{j = 1}^d x_i \ot e_j \ot e_j + e_j \ot x_i \ot e_j  + e_j \ot e_j \ot x_i$, with $e_j$ the $j$-th standard basis vector in $\R^{d}$.
Here, $\Nscr_2 \cup \Nscr_3$ is a partition of the set of $N$ trajectories into two disjoint sets of size $N_2$ and $N_3$ respectively, which enables us to obtain independent estimates of the matrix $M_2$ from $\Nscr_2$ and of the third order tensor $M_3$ from $\Nscr_3$. 

Let $y_i:= \ytilde_i - \xi_i$ be ``cleaned'' observations. If $\set{(y_i, x_i)}_{i \in [N]}$ were observed, we would be solving a mixture of linear regressions with i.i.d. noise $\eta_i$ and no perturbations $\xi_i$: $y_i = \brk{\beta_{k_i}, x_i} + \eta_i.$
We define the moments estimated with unperturbed $y_i$: 
\begin{align}\label{eq:M23hat-MLR}
\Mhat_2 &= \frac{1}{2N_2} \sum_{i\in \Nscr_2} y_i^2 (x_i \ot x_i - I_d)\ \text{ and }\
\Mhat_3 = \frac{1}{6N_3} \sum_{i\in \Nscr_3} y_i^3 \lrp{x_i^{\ot 3} - \Escr(x_i)}.
\end{align}
It can be verified by multiple applications of Stein's identity \citep{janzamin2014score} using that if $x_i$ is isotropic gaussian and uncorrelated with the zero-mean noise $\eta_i$, then $\Mhat_2$ and $\Mhat_3$ are unbiased estimators of the two mixtures of moment tensors:
\begin{align}
\E[\Mhat_2] &= M_2\coloneq \sum_{k = 1}^K p_k \beta_k \ot \beta_k\ \text{ and }\
\E[\Mhat_3] = M_3\coloneq \sum_{k = 1}^K p_k \beta_k \ot \beta_k \ot \beta_k.
\end{align}
Empirical estimates $\Mtilde_2$ and $\Mtilde_3$ differ from the unbiased estimators $\Mhat_2$ and $\Mhat_3$ only by factors involving the perturbations $\xi_i$. When $\xi_i$ have bounded norm, as it is in \eqref{eq:xi}, %
then $\Mtilde_2$ and $\Mtilde_3$ provide good estimates of the target moments $M_2$ and $M_3$. 

\paragraph{Whitening factors.} 
Although it is possible to run tensor decomposition on the original estimates of $M_3$, a useful intermediate step is to whiten the set of $d$-dimensional tensor factors $\set{\beta_k}_{k \in [K]}$ by projecting them onto the $K$-dimensional subspace of $\R^d$ spanned by the factors themselves, to yield a set of orthonormal $K$-dimensional vectors $\set{W'\beta_k}_{k \in [K]}$. Here, $W\in \R^{d \by K}$ is a whitening matrix derived from the singular value decomposition (SVD) of the moment matrix $M_2 = \sum_{k=1}^K p_k \beta_k \ot \beta_k$. This whitening step is a form of dimensionality reduction; estimating and decomposing the whitened third-order tensor $\clean = \sum_{k=1}^K p_k (W' \beta_k)^{\ot 3} \in (\R^K)^{\ot 3}$ has lower computational and statistical demands than for the original $M_3$. Additionally, since the transformed factors $W'\beta_k$ have unit norm, it is possible to disentangle the scaling between $p_k$ and $\beta_k$ through the dewhitening step (c.f., Lines \ref{line:dewhiten1}-\ref{line:dewhiten2} in Algorithm \ref{alg:MLR-estimator}). Finally, if the number of mixture components $K$ were unknown, $K$ could also be estimated the SVD of empirical estimates of $M_2$. 

In Algorithm \ref{alg:MLR-estimator}, the estimated whitening matrix $\Wtilde \in \R^{d\by K}$ is obtained from the SVD of the estimate $\Mtilde_2$, such that $\Wtilde'\Mtilde_2 \Wtilde = I_K$. %
The whitened third order tensor $\tiltil$, which estimates $M_3^W$, then has orthonormal $K$-dimensional components, making it amenable to the standard robust tensor power iteration method (\cite{anandkumar2014tensor}) for orthogonal tensor decomposition. In the last step, the output $\set{(\tilde{w}_k, \tilde{\beta}_k)}_{k \in [K]}$ of the decomposition is dewhitened using $\Wtilde$ to return estimates $\phat_k$ and $\betahat_k$ of the original mixture weights and coefficients, for each component $k \in [K]$.

\IncMargin{1em}%
\begin{algorithm2e}[bhtp] %
\caption{Mixture of Linear Regressions Estimator}\label{alg:MLR-estimator}
\SetAlgoLined
\SetKwFor{For}{for}{}{} 
\LinesNumbered
\DontPrintSemicolon
\KwIn{%
\begin{minipage}[t]{0.98\linewidth}
    $\set{(x_i, y_i) \in \R^d\by \R \mid i \in [N]}$ --- Regression data \\
    $\Nscr_2 \cup \Nscr_3 = [N]$ --- Index set partition for estimating moments \\
    $K$ --- Number of mixture components
  \end{minipage}
}
\KwOut{$\set{(\phat_k, \betahat_k)}_{k \in [K]}$ --- Estimated mixture parameters and weights}

\textbf{Whitening:} \\
$\Mtilde_2 \gets \frac{1}{2N_2}\sum_{i \in \Nscr_2} \lrb{y_i^2 (x_i^{\ot 2} - I_{d})}$ \tcp*[r]{2nd order tensor}
$U\Sigma U^T \gets \text{SVD}(\Mtilde_2, K)$ \tcp*[r]{Rank-$K$ approximation}
$\Wtilde \gets U\Sigma^{-1/2}$ \tcp*[r]{Whitening matrix}

\textbf{Tensor estimation and decomposition:} \\
$\Mtilde^{\Wtilde}_3 \gets \frac{1}{6N_3}\sum_{i \in \Nscr_3} \left[ y_{i}^3 (\Wtilde'x_i)^{\ot 3} - \Escr(x_i)(\Wtilde, \Wtilde, \Wtilde)\right]$ \tcp*[r]{3rd order tensor}

$\set{(\ptilde_k, \betatilde_k)\in \R\by \R^{K} \mid k \in [K]} \gets$ Orthogonal Tensor Decomposition$\left(\frac{1}{6}\Mtilde^{\Wtilde}_3, K\right)$ \\
\For{$k \in [K]$ \textnormal{\textbf{:}} \label{line:dewhiten1}}
{%
    $\phat_k \gets 1/\ptilde_k^2, \quad \betahat_k \gets \ptilde_k \pseudo{(\Wtilde')}\betatilde_k$ \tcp*[f]{Dewhiten} \label{line:dewhiten2}
}

\end{algorithm2e}
\DecMargin{1em} %

\section{Analysis}\label{sec:analysis}
Proposition \ref{prop:MLDS-result} provides our main result of finite sample error bounds for learning MLDS via Algorithm \ref{alg:MLDS-estimator}. 
Using the mixtures of linear regression subroutine, we essentially run tensor decomposition on a whitened (orthonormal) version of the third-order tensor
\begin{align}
\frac{L}{6NT} \sum_{i\in [N]} \sum_{t \in\Jscr} \bigg(y_{i,t}^3 \ubar_{i,t}^{\ot 3} -  y_{i,t}^3 \sum_{k\in[mL]} (e_k \ot \ubar_{i,t} \ot e_k + \ubar_{i,t} \ot e_k \ot e_k + e_k \ot e_k \ot \ubar_{i,t})\bigg),%
\end{align}
which estimates $\sum_{k=1}^K p_k (\gscr_k)^{\ot 3}$. Here, $e_k$ is the $k$-th standard basis vector in $\R^{mL}$. 
\begin{proposition}\label{prop:MLDS-result}
Let data $\Dscr = \set{(u_{i,t-1}, y_{i,t}) \mid i \in [N], t \in [T]}$ be generated from a mixture of linear dynamical systems with parameters $ \set{(p_k, \gscr_k)}_{k \in [K]}$, and let $L$ be an integer such that $1 \leq L \leq T$. 
Let $\set{(\phat_k, \widehat{\gscr}\uL_k)}_{k \in [K]}$ be the estimated mixture parameters obtained from running Algorithm \ref{alg:MLDS-estimator} on the data $\Dscr$.
Let $\sigma_y^2\coloneq (\sigma_u^2+ \sigma_w^2)\Gamma_{\max}$. For any $\eps>0, \delta \in (0,1)$, when
\begin{align}
N_2 T &\gsim \frac{\sigma_y^4L\Gamma_{\max}^3}{\eps^2\pmin^2} \cdot \lrp{\sigma_K^5  \ln^4 \lrp{\frac{N_2 T\cdot 9^{Lm}}{\delta L}}\ln\lrp{\frac{9^{Lm}}{\delta}} +  \frac{\delta}{9^{Lm}}\cdot \frac{\sigma_y^4\Gamma_{\max}^3}{\eps^2\pmin^2}}, \\
N_3 T &\gsim \frac{ \sigma_y^6 L}{\eps^2\pmin^2\sigma_K^3} \cdot \lrp{ \ln^6 \lrp{\frac{33^K\cdot N_3 T}{\delta L}}\ln\lrp{\frac{33^K}{\delta}}+\frac{\delta}{33^K}\cdot \frac{ \sigma_y^6}{\eps^2\pmin^2\sigma_K^3}}, \\
\eps &\lsim  \frac{\sigma_y^3}{\sigma_K^{3/2} \pmin},\ \text{and} \\
L &\geq \ln \lrp{\frac{\sigma_y^4 \Gamma_{\max}}{\eps^2 \pmin^2 \sigma_K^3}\cdot \frac{ C_{\rho} \rho}{1-\rho} \lrb{\Gamma_{\max}^{3/2} \ln^2\lrp{\frac{9^{Lm} N_2}{\delta}} + \sigma_y \ln^3 \lrp{\frac{33^K N_3}{\delta}}}}\cdot \frac{1}{\ln(1/\rho)},
\end{align}
it holds that with probability at least $1-\delta$, %
there exists a permutation $\pi: [K] \to [K]$ such that
\begin{align}\label{eq:MLDS-err-bnd}
\opnorm{\widehat{\gscr}\uL_{\pi(k)} - \gscrL_k} \leq \eps \cdot \frac{\sigma_K^{1/2}}{\pmin^{3/2}} , \quad \abs{\phat_{\pi(k)} - p_k} < \eps p_k^{3/2}, \quad \text{for } k \in [K].
\end{align}
\end{proposition}

The proof of Proposition \ref{prop:MLDS-result} is found in Section \ref{sec:MLDS-proof}. %
It proceeds by first bounding estimation error for MLR with noise and perturbations (see Theorem~\ref{thm:perturbed-MLR}), and then translates those bounds to estimation error in Markov parameters for MLDS. In more detail, we bound first the deviation of $\Mtilde_2$ from $M_2$, then the estimation error of the whitening matrix $\Wtilde$ derived from $\Mtilde_2$, and finally the estimation error of $\tiltil$ from $\clean$. In each step, we use various concentration results and control the effect of the perturbations $\xi_i$.
Note that $\Mtilde_2$ and $\Mtilde_3$ involve 4th and 6th order moments of the regressor random variables which are effectively $d$- and $K$-dimensional, respectively, leading to polynomial dependence on the dimensions and variance parameters.

Next, results on the robustness of the tensor power method \citep{anandkumar2014tensor} are applied to transfer bounds on $\opnorm{\tiltil - \clean}$ to the orthonormalized components of the tensor $\tiltil$, i.e., the whitened projections of the mixture components and their corresponding mixture weights $\set{(\ptilde_k, \betatilde_k)}_{k \in [K]}$.
The estimation error of the whitened mixture components is propagated through a dewhitening step, yielding estimates for $\set{(\phat_k, \betahat_k)}_{k \in [K]}$.
We obtain Proposition \ref{prop:MLDS-result} by adapting this analysis to estimating $\set{(\phat_k, \widehat{\gscr}\uL_k)}_{k \in [K]}$ from input-output data.

\paragraph{Interpretation of results.}
Let us rewrite above sample complexity results in Proposition \ref{prop:MLDS-result} in terms of upper bounds on estimation error, for simplicity setting $N_2 = N_3$, ignoring log factors, and keeping only the dependence on $N, T, L$ and $\rho$. Then we have that given $N, T$, and $L$, the estimation error $\eps$ of the $L$ impulse response parameters of the mixture components roughly scales as 
\begin{align}
    \eps \gsim \frac{L^3}{\sqrt{NT}} + L\rho^{L/2}.
\end{align}
The first term of the error decreases as $1/\sqrt{NT}$ which is to be expected from solving a linear regression with a sample size of $NT$. However, to circumvent the dependency structure in the covariates $\ubar_{it}$, we take every $L$th sample of each trajectory, cutting the effective sample size to $NT/L$. Furthermore, as we increase $L$, the dimension $Lm$ of the estimated parameters increases, which enters polynomially into the estimation error bound. The second term of the error, $L\rho^{L/2}$, is due to the tail of the truncated impulse response sequence, corresponding to the perturbations $\xi_i$ in the MLR model (Definition \ref{def:MLR}), and decreases exponentially with $L$ for stable systems. By growing $L$ at a rate of $(NT)^{1/(6+a)}$ as $NT \to \infty$, for $a>0$, the two components of the estimation error asymptotically decrease to zero.

A particularly interesting property of the tensor decomposition approach to mixtures of linear systems, and which arises in other cases of learning multiple models with latent structure, is the tradeoff between $N$ and $T$ in finite sample error bounds. The setting of observing just a few trajectories, but where each trajectory is long (small $N$, large $T$), may yield the same estimation error levels as the setting of observing many short trajectories (large $N$, small $T$) from the mixture. The flexibility in sample complexity from assuming and learning a latent structure can prove useful in a wide range of data sets with varying compositions of individual versus collective sample sizes.

\section{Simulations}\label{sec:simulations}
\begin{figure}[htbp]
    \floatconts
    {fig:bigfig1}
    {\caption{\vspace{-1em}Results for estimating the first $L=7$ Markov parameters of $K=3$ mixture components. (a) Average parameter estimation error vs. $N$ for various $T$. Standard errors for 15 trials shown. (b) Level sets of estimation error as a function of $N$ and $T$.}}
    {%
    \subfigure[]{
        \includegraphics[width=0.48\textwidth]{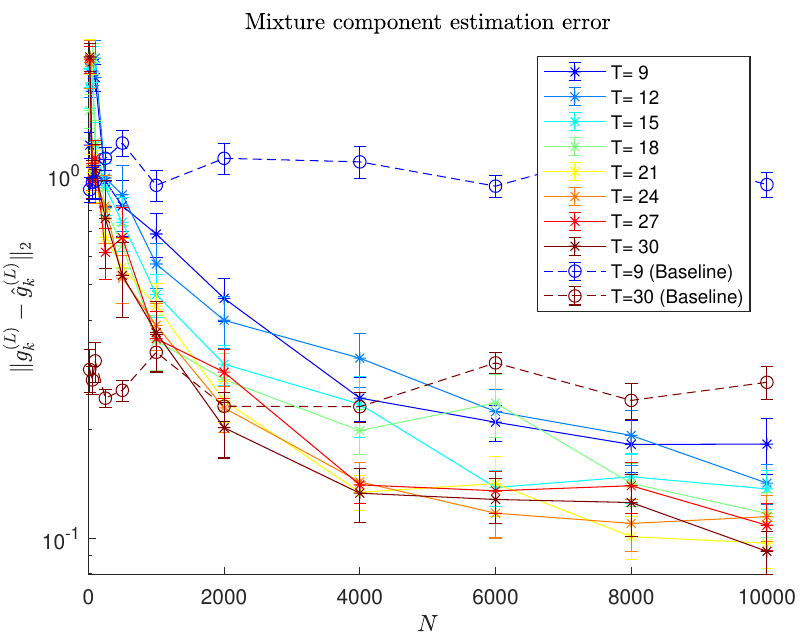}
    }
    \subfigure[]{%
        \includegraphics[width=0.48\textwidth]{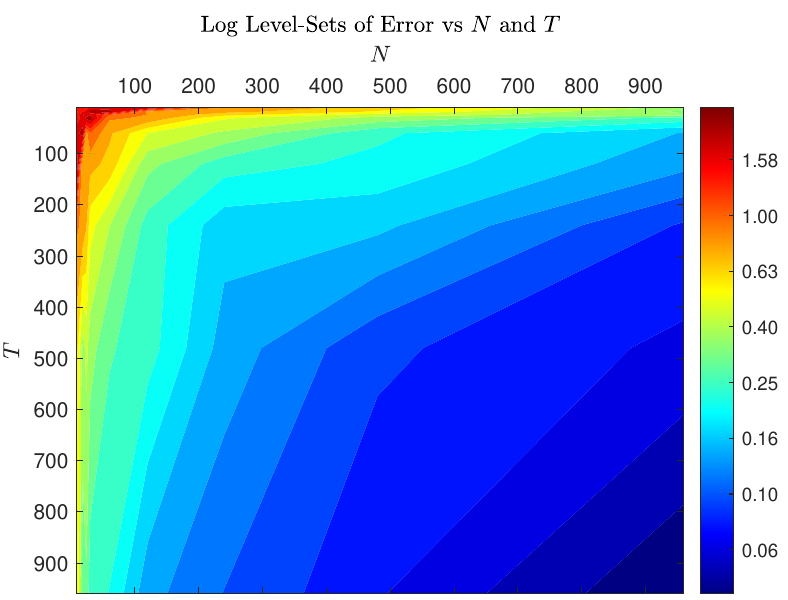}
    }%
    }
\end{figure}

We evaluate the performance of Algorithm \ref{alg:MLDS-estimator} in estimating mixtures of linear systems through a series of simulations. 
In each trial, $K=3$ single-input single-output linear models of order $n=3$ were generated, with spectral radii varying between 0.6 and 0.9. $N$ unlabeled trajectories of length $T$ were sampled from the resulting mixture of $K$ models, for $T\in [9,960]$ and $N\in [9,10,000]$. The first $L=7$ Markov parameters of each mixture component were estimated.

Figure \ref{fig:bigfig1} plots the estimation error $(1/K)\sum_{k=1}^K \big\lVert\gscr_k\uL - g_k\uL\big\rVert_2$ for Algorithm \ref{alg:MLDS-estimator}, both (a) as a function of $N$ for various $T$, and (b) as level sets on the $(N,T)$ plane. Additionally, in Figure \ref{fig:bigfig1}(a), we plot for comparison the error of the ``baseline estimator'' which estimates Markov parameters individually for each observed trajectory using ordinary least squares \citep{oymak2021revisiting}. The error for the baseline estimator is calculated as $(1/N) \sum_{i=1}^N \big\lVert\gscr_{k_i}\uL - \hat{g}_i\uL \big\rVert_2$. 
Although the tensor approach initially has higher error in the small $N$ regime, likely due to the use of higher-order moments, it is able to leverage shared structure across $N$ trajectories to achieve lower estimation error for larger $N$ versus the baseline estimates. %
This effect is particularly apparent for smaller $T$, which is a common regime in practical applications. 

Figure \ref{fig:bigfig1}(b) further shows how the performance of the MLDS estimator improves with both $N$ and $T$. Empirically, we find that the tensor decomposition approach is quite sensitive to the conditioning of the matrix $M_2$ of Markov parameters, which is related to the degree of non-degeneracy of the mixture parameters. %
In particular, the norm of the whitening matrix $W$ depends on the smallest singular value of the estimated $M_2$, which affects the downstream estimation of the whitened third-order tensor $M_3^W$ and the accuracy of the final dewhitened estimates.
For future work, it would be interesting to combine the MLDS estimator with iterative mixture estimation methods, which may improve the accuracy and sample complexity of the approach.

\bibliography{dynamics}

\begin{thebibliography}{21}
\providecommand{\natexlab}[1]{#1}
\providecommand{\url}[1]{\texttt{#1}}
\expandafter\ifx\csname urlstyle\endcsname\relax
  \providecommand{\doi}[1]{doi: #1}\else
  \providecommand{\doi}{doi: \begingroup \urlstyle{rm}\Url}\fi

\bibitem[Anandkumar et~al.(2014)Anandkumar, Ge, Hsu, Kakade, Telgarsky,
  et~al.]{anandkumar2014tensor}
Animashree Anandkumar, Rong Ge, Daniel~J Hsu, Sham~M Kakade, Matus Telgarsky,
  et~al.
\newblock Tensor decompositions for learning latent variable models.
\newblock \emph{J. Mach. Learn. Res.}, 15\penalty0 (1):\penalty0 2773--2832,
  2014.

\bibitem[Bader et~al.(2023)Bader, Kolda, et~al.]{tensortoolbox}
Brett~W. Bader, Tamara~G. Kolda, et~al.
\newblock Tensor toolbox for matlab, version 3.6, 9 2023.
\newblock URL \url{https://www.tensortoolbox.org/}.

\bibitem[Bakshi et~al.(2023{\natexlab{a}})Bakshi, Liu, Moitra, and
  Yau]{bakshi2023new}
Ainesh Bakshi, Allen Liu, Ankur Moitra, and Morris Yau.
\newblock A new approach to learning linear dynamical systems.
\newblock \emph{arXiv preprint arXiv:2301.09519}, 2023{\natexlab{a}}.

\bibitem[Bakshi et~al.(2023{\natexlab{b}})Bakshi, Liu, Moitra, and
  Yau]{bakshi2023tensor}
Ainesh Bakshi, Allen Liu, Ankur Moitra, and Morris Yau.
\newblock Tensor decompositions meet control theory: learning general mixtures
  of linear dynamical systems.
\newblock In \emph{International Conference on Machine Learning}, pages
  1549--1563. PMLR, 2023{\natexlab{b}}.

\bibitem[Chen and Poor(2022)]{chen2022learning}
Yanxi Chen and H~Vincent Poor.
\newblock Learning mixtures of linear dynamical systems.
\newblock In \emph{International Conference on Machine Learning}, pages
  3507--3557. PMLR, 2022.

\bibitem[Ernst et~al.(2005)Ernst, Nau, and Bar-Joseph]{ernst2005clustering}
Jason Ernst, Gerard~J Nau, and Ziv Bar-Joseph.
\newblock Clustering short time series gene expression data.
\newblock \emph{Bioinformatics}, 21\penalty0 (suppl\_1):\penalty0 i159--i168,
  2005.

\bibitem[Gragg and Lindquist(1983)]{gragg1983partial}
William~B Gragg and Anders Lindquist.
\newblock On the partial realization problem.
\newblock \emph{Linear Algebra and its Applications}, 50:\penalty0 277--319,
  1983.

\bibitem[Ho and K{\'a}lm{\'a}n(1966)]{ho1966effective}
BL~Ho and Rudolf~E K{\'a}lm{\'a}n.
\newblock Effective construction of linear state-variable models from
  input/output functions: Die konstruktion von linearen modeilen in der
  darstellung durch zustandsvariable aus den beziehungen f{\"u}r ein-und
  ausgangsgr{\"o}{\ss}en.
\newblock \emph{at-Automatisierungstechnik}, 14\penalty0 (1-12):\penalty0
  545--548, 1966.

\bibitem[Janzamin et~al.(2014)Janzamin, Sedghi, and
  Anandkumar]{janzamin2014score}
Majid Janzamin, Hanie Sedghi, and Anima Anandkumar.
\newblock Score function features for discriminative learning: Matrix and
  tensor framework.
\newblock \emph{arXiv preprint arXiv:1412.2863}, 2014.

\bibitem[Kolda and Bader(2009)]{kolda2009tensor}
Tamara~G Kolda and Brett~W Bader.
\newblock Tensor decompositions and applications.
\newblock \emph{SIAM review}, 51\penalty0 (3):\penalty0 455--500, 2009.

\bibitem[Kozyakin(2009)]{kozyakin2009accuracy}
Victor Kozyakin.
\newblock On accuracy of approximation of the spectral radius by the gelfand
  formula.
\newblock \emph{Linear Algebra and its Applications}, 431\penalty0
  (11):\penalty0 2134--2141, 2009.

\bibitem[Li and Liang(2018)]{li2018learning}
Yuanzhi Li and Yingyu Liang.
\newblock Learning mixtures of linear regressions with nearly optimal
  complexity.
\newblock In \emph{Conference On Learning Theory}, pages 1125--1144. PMLR,
  2018.

\bibitem[Modi et~al.(2021)Modi, Faradonbeh, Tewari, and
  Michailidis]{modi2021joint}
Aditya Modi, Mohamad Kazem~Shirani Faradonbeh, Ambuj Tewari, and George
  Michailidis.
\newblock Joint learning of linear time-invariant dynamical systems.
\newblock \emph{arXiv preprint arXiv:2112.10955}, 2021.

\bibitem[Oymak and Ozay(2021)]{oymak2021revisiting}
Samet Oymak and Necmiye Ozay.
\newblock Revisiting ho--kalman-based system identification: Robustness and
  finite-sample analysis.
\newblock \emph{IEEE Transactions on Automatic Control}, 67\penalty0
  (4):\penalty0 1914--1928, 2021.

\bibitem[Sarkar et~al.(2021)Sarkar, Rakhlin, and Dahleh]{sarkar2021finite}
Tuhin Sarkar, Alexander Rakhlin, and Munther~A Dahleh.
\newblock Finite time lti system identification.
\newblock \emph{The Journal of Machine Learning Research}, 22\penalty0
  (1):\penalty0 1186--1246, 2021.

\bibitem[Vershynin(2018)]{vershynin2018high}
Roman Vershynin.
\newblock \emph{High-dimensional probability: An introduction with applications
  in data science}, volume~47.
\newblock Cambridge university press, 2018.

\bibitem[Yi et~al.(2014)Yi, Caramanis, and Sanghavi]{yi2014alternating}
Xinyang Yi, Constantine Caramanis, and Sujay Sanghavi.
\newblock Alternating minimization for mixed linear regression.
\newblock In \emph{International Conference on Machine Learning}, pages
  613--621. PMLR, 2014.

\bibitem[Yi et~al.(2016)Yi, Caramanis, and Sanghavi]{yi2016solving}
Xinyang Yi, Constantine Caramanis, and Sujay Sanghavi.
\newblock Solving a mixture of many random linear equations by tensor
  decomposition and alternating minimization.
\newblock \emph{arXiv preprint arXiv:1608.05749}, 2016.

\bibitem[Zhang et~al.(2023)Zhang, Kang, Lee, Tomlin, Levine, Tu, and
  Matni]{zhang2023multi}
Thomas~T Zhang, Katie Kang, Bruce~D Lee, Claire Tomlin, Sergey Levine, Stephen
  Tu, and Nikolai Matni.
\newblock Multi-task imitation learning for linear dynamical systems.
\newblock In \emph{Learning for Dynamics and Control Conference}, pages
  586--599. PMLR, 2023.

\bibitem[Zheng and Li(2020)]{zheng2020non}
Yang Zheng and Na~Li.
\newblock Non-asymptotic identification of linear dynamical systems using
  multiple trajectories.
\newblock \emph{IEEE Control Systems Letters}, 5\penalty0 (5):\penalty0
  1693--1698, 2020.

\bibitem[Zhong et~al.(2016)Zhong, Jain, and Dhillon]{zhong2016mixed}
Kai Zhong, Prateek Jain, and Inderjit~S Dhillon.
\newblock Mixed linear regression with multiple components.
\newblock \emph{Advances in neural information processing systems}, 29, 2016.

\end{thebibliography}

\newpage
\appendix
\section{APPENDIX}
\subsection{Mixtures of Linear Regression Recovery Results}\label{sec:thm-perturbed-MLR}
In this section, we present $(\eps, \delta)$-PAC learnability conditions for the problem of learning mixtures of perturbed linear regressions (as described in Definition \ref{def:MLR}. For the following result, we assume there exists a constant $b>0$ such that $\enorm{\beta_k} \leq b$ for all $k \in [K]$, and we define $\sigma_y^2 \coloneq b^2 + \sigma_{\eta}^2$. To simplify error bound expressions, we assume without loss of generality that $\sigma_{\xi} \leq 1$ (otherwise we can normalize data and update the value of $b$). Finally, we assume that the minimum mixture weight is lower bounded, i.e., $\pmin = \min_{k \in [K]} p_k > 0$ and that $\sigma_K(M_2) > 0$, mirroring the mixture assumptions in Section \ref{sec:Assumptions} for the MLDR model. 
\begin{theorem}\label{thm:perturbed-MLR}
Let $\set{(\betahat_k, \phat_k)\mid k \in [K]}$ be estimates obtained from running Algorithm \ref{alg:MLR-estimator} given data $\set{(x_i, \ytilde_i) \in \R^d \by \R}_{i \in [N]}$ generated from the MLR model in Definition \ref{def:MLR}.
Let $\sigma_K:= \min(\sigma_K(M_2),1)$.
For any $\eps>0, \delta \in (0,1)$, suppose the hyperparameters $(\Riter, \Rstart)$ in the subroutine Algorithm \ref{alg:tensor-power-method} satisfy \eqref{eq:alg-param-values}.
When the following conditions are satisfied: 
\begin{align}
N_2 &\gsim \max \set*{\frac{\sigma_y^4\sigma_K^5 \opnorm{M_3}^2}{\eps^2\pmin^2} \ln^4 \lrp{\frac{N_2\cdot 9^{d}}{\delta}}\ln\lrp{\frac{9^{d}}{\delta}}, \frac{\delta}{9^{d}}\cdot \frac{\sigma_y^8\opnorm{M_3}^4}{\sigma_K^{10}\eps^4\pmin^4}}, \\
N_3 &\gsim \max \set*{ \frac{ \sigma_y^6}{\eps^2\pmin^2\sigma_K^3} \ln^6 \lrp{\frac{33^K\cdot N_3}{\delta}}\ln\lrp{\frac{33^K}{\delta}},\frac{\delta}{33^K}\cdot \frac{\sigma_y^{12}}{\eps^4\pmin^4\sigma_K^6}}, \\
\sigma_\xi &\lsim \frac{\eps\pmin \sigma_K^3}{\sigma_y} \min\set*{\frac{1}{\opnorm{M_3} \ln \lrp{\frac{9^{d} N_2}{\delta}}\ln\lrp{\frac{N_2}{\delta}}} , \frac{1}{\sigma_y \ln^{3/2} \lrp{\frac{33^K N_3}{\delta}}\ln^{3/2} \lrp{\frac{N_3}{\delta}}}}\\
\eps &< \frac{\sigma_y^3}{1.55\sigma_K^{3/2} \pmin},
\end{align}%
then $\Pr{\opnorm{\tiltil - \clean} \lsim \eps} \geq 1-\delta$, and 
there exists a permutation $\pi: [K] \to [K]$ such that for all $k \in [K]$, 
\begin{align}
\opnorm{\betahat_{\pi(k)} - \beta_k} \leq \eps \cdot \frac{\sigma_K^{1/2}}{\pmin^{3/2}} , \quad \abs{\phat_{\pi(k)} - p_k} < \eps p_k^{3/2}
\end{align}
with probability at least $1-\delta$.
\end{theorem}

\begin{proof}
\setlength{\parindent}{15pt}
We first bound the estimation error of $\Mtilde_2$ for $M_2$, then propagate this error through to the whitening matrix $\Wtilde$ and then to the estimation error of $\tiltil$ for $\clean$. We then apply a standard robustness result for orthogonal tensor decomposition to obtain estimation error bounds for the individual components and weights of the mixture. 

\vspace{1em}
\noindent \textbf{Estimating M2.}
Recall the definitions of $\Mtilde_2$ and $\Mhat_2$ in \eqref{eq:M23tilde-MLR} and \eqref{eq:M23hat-MLR}, respectively. 
By the triangle inequality, we decompose the $M_2$ estimation error as 
\begin{align}\label{eq:triangle-perturb-M2}
\opnorm{\Mtilde_2 - M_2} \leq \opnorm{\Mtilde_2 - \Mhat_2} + \opnorm{\Mhat_2 - M_2}.
\end{align}
Let $\eps_2 = \eps \sigma_K^{5/2} /\opnorm{M_3}$.
By Lemma \ref{lem:M2-perturb}, when 
\begin{align}\label{eq:pf-xi-M2-cond}
\sigma_{\xi} \lsim \frac{\eps_2}{\sigma_y \ln \lrp{\frac{N_2}{\delta}} \ln\lrp{\frac{9^d N_2}{\delta}}}, 
\end{align}
we have that $\opnorm{\Mtilde_2 - \Mhat_2} \leq \eps_2$ with probability at least $1-\delta$. Essentially, if $\xi_i$ is small enough with high probability, then estimating the moment from the perturbed observations $\ytilde_i$ is not too far from estimating the moment based on the unperturbed (but noisy) samples $y_i$. 

Next, by Corollary \ref{cor:M2-conc}, when $N_2$ satisfies condition \eqref{eq:N2-condition} in Corollary \ref{cor:M2-conc} with $(\eps_2, \delta)$, when $\eps_2 < \sigma_y^2/1.51$ (to simplify the expressions)
\begin{align}\label{eq:pf-N2-cond-copy}
N_2 \gsim \max \set*{\frac{\sigma_y^4}{\eps_2^2} \ln^4 \lrp{\frac{N_2\cdot 9^d}{\delta}}\ln\lrp{\frac{\cdot 9^d}{\delta}}, \frac{\delta}{9^d}\lrp{\frac{\sigma_y^8}{\eps_2^4}}} 
\end{align}
with probability at least $1-\delta$, $\opnorm{M_2 - \Mhat_2} \leq \eps_2$. 
In total, under conditions \eqref{eq:pf-xi-cond} and \eqref{eq:pf-N2-cond-copy}, $\opnorm{\Mtilde_2 - \Mhat_2} \leq \eps_2$.

Let us now impose the condition that
\begin{align}\label{eq:pf-eps3-cond}
\eps < \frac{\opnorm{M_3}}{3\sigma_K^{3/2}},
\end{align}
so that $\eps_2 < \sigma_K/3$. Then by Lemma \ref{lem:whiten-perturb}, $\opnorm{\Wtilde} \leq 2\opnorm{W}\leq 2\sigma_K^{-1/2} $ and by Corollary \ref{cor:whitened}, 
\begin{align}\label{eq:pf-whitened}
\opnorm{\mixedtil - \clean} &\lsim  \frac{\opnorm{\Mtilde_2-M_2}}{\sigma_K(M_2)^{5/2}} \opnorm{M_3} \lsim \eps.
\end{align}
Note that we can bound $\opnorm{M_3}$ by $b^3$, and also that $\sigma_y^3\geq b^3 \geq  \opnorm{M_3}$.
Additionally, we note that Lemma \ref{lem:whiten-perturb} implies that $\opnorm{\pseudo{\Wtilde}} \leq 2\opnorm{\pseudo{W}}=2\sigma_K^{1/2}$ and 
\begin{align}
    \opnorm{\pseudo{W} - \pseudo{\Wtilde}} &\leq 2\opnorm{\pseudo{W}}\opnorm{M_2 - \Mtilde_2}/\sigma_K(M_2) \\
    & = 2\sigma_K^{-1/2}\eps_2 \\
    & = 2\eps \sigma_K^2/\opnorm{M_3}.
\end{align}
These bounds will be used in the dewhitening part of the analysis. 

\vspace{1em}
\noindent \textbf{Estimating $M_3$.}
Again by the triangle inequality, we decompose the estimation error of the third order whitened tensor $\clean$: 
\begin{align}
\opnorm{\tiltil - \clean}&\leq \opnorm{\tiltil-\Mhat_3^{\Wtilde}} + \opnorm{\Mhat_3^{\Wtilde} - \mixedtil} + \opnorm{\mixedtil - \clean}.
\end{align}
The first term on the right hand side of the inequality captures the effects of the perturbations $\xi_i$ on the estimate, the second term captures the standard empirical moment estimation of a third order tensor, and the third term captures the effects of using an estimated whitening matrix rather than the true one. 

By Lemma \ref{lem:tensor-factor}, the first term can be decomposed as $\opnorm{\tiltil - \Mhat_3^{\Wtilde}} \leq \opnorm{\Mtilde_3 - \Mhat_3}\opnorm{\Wtilde}^3$. Combining this with Lemma \ref{lem:M3-perturb}, with $V = \Wtilde$, to control the effect of the perturbations $\xi_i$ on the $M_3$ estimate, it holds that when 
\begin{align}\label{eq:pf-xi-M3-cond}
\sigma_\xi \leq \frac{\eps}{\sigma_y^2 \opnorm{\Wtilde}^6 \ln^{3/2} \lrp{\frac{N_3 33^K}{\delta}}\ln^{3/2} \lrp{\frac{N_3}{\delta}}},
\end{align} 
we have $\Pr{\opnorm{\tiltil - \Mhat_3^{\Wtilde}} \geq \eps} \leq 1-\delta$.

The second term in the inequality is bounded by concentration results for the empirical third order moment $\Mhat_3$ evaluated on the empirical whitening matrix $\What$. When $N_3$ satisfies \eqref{eq:N3-condition} in Corollary \ref{cor:tensor-conc} with $(\eps, \delta)$ and $V = \Wtilde$, then $\opnorm{\Mhat_3^{\Wtilde} - \mixedtil} \leq \eps$ with probability at least $1-\delta$.
Applying the bound $\opnorm{\Wtilde} \leq 2\sigma_K^{-1/2}$ from the $M_2$ analysis, the condiiton on $N_3$ becomes: 
\begin{align}\label{eq:pf-N3-cond}
N_3 \gsim \max \set*{ \frac{ \sigma_y^6}{\eps^2 \sigma_K^3} \ln^6 \lrp{\frac{33^K\cdot N}{\delta}}\ln\lrp{\frac{\cdot 33^K}{\delta}},\frac{\delta}{33^K}\frac{\sigma_y^{12}}{\eps^4\sigma_K^6}}.
\end{align}
where we additionally impose the benign condition that $\eps < \sigma_y^3/(1.55\sigma_K^{3/2})$ to simplify this condition. 
The third term, capturing the effect of the whitening matrix estimation error, is controlled in \eqref{eq:pf-whitened}. 

Thus, under the combined conditions of \eqref{eq:pf-xi-M2-cond}, \eqref{eq:pf-N2-cond-copy}, \eqref{eq:pf-eps3-cond}, \eqref{eq:pf-xi-M3-cond}, \eqref{eq:pf-N3-cond} and the assumptions on $\eps$, which can be simplified to $\eps < \opnorm{M_3}\sigma_K^{-3/2}$, we have that 
\begin{align}
\Pr{\opnorm{\tiltil - \clean} \lsim \eps} \geq 1-\delta.
\end{align}

\noindent \textbf{Tensor Decomposition.}
We now propagate this tensor estimation bound through robustness results for orthogonal tensor decomposition and through dewhitening the tensor components to obtain bounds on estimating the individual components and weights of the mixture of linear regressions.
Since we can decompose the true whitened tensor as 
\begin{align}
\clean = \sum_{k=1}^N \frac{1}{p_k} (\sqrt{p_k}W'\beta_k)^{\ot 3}
\end{align}
where it can be shown \citep{anandkumar2014tensor} that $\set{\sqrt{p_k}W'\beta_k}_{k \in [K]}$ are orthonormal, running the orthogonal tensor decomposition in Algorithm \ref{alg:tensor-power-method} gives us estimates $\set{(\betatilde_k, \ptilde_k)}_{k \in [K]}$ of $\set{(\sqrt{p_k}W'\beta_k, 1/p_k)}_{k \in [K]}$. 
By Lemma \ref{lem:tensor-robustness}, whenever
\begin{align}\label{eq:pf-eps-tensor}
\eps \lsim \pmin/K,
\end{align}  
and given $\delta$, for $(\Riter, \Rstart)$ in Algorithm \ref{alg:tensor-power-method} satisfying
\begin{align}
\Riter \lsim \ln(K) + \ln\ln\lrp{\frac{1}{\eps}}, \quad \Rstart \gsim \poly(K)\ln(1/\delta), 
\end{align}
with probability at least $1-\delta$, there exists a permutation $\pi:K \to K$ such that for all $k \in [K]$, 
\begin{align}\label{eq:pf-tensor-robust}
\enorm{\betatilde_{\pi(k)} - \sqrt{p_k}W'\beta_k} \lsim \eps/p_k , \quad 
\abs{\ptilde_{\pi(k)} -\frac{1}{\sqrt{p_k}}} \lsim \eps.
\end{align} 

\noindent \textbf{Dewhitening.} 
We now dewhiten the output of the orthogonal tensor decomposition to produce estimates. For simplicity let us assume the permutation $\pi$ in \eqref{eq:pf-tensor-robust} is the identity. Then let us define the dewhitened mixture component estimates as
\begin{align}
\phat_k = \frac{1}{\ptilde_k^2}, \quad \betahat_k = \ptilde_k \pseudo{(\What')} \betatilde_k.
\end{align}
To propagate the estimation error through the dewhitening process, note that 
\begin{align}
\enorm{\betahat_k - \beta_k} &= \opnorm{\ptilde_k\pseudo{(\Wtilde')}\lrp{\betatilde_k - \sqrt{p_kW'\beta_k}} + \lrp{\ptilde_k \pseudo{(\Wtilde')} - \frac{1}{\sqrt{p_k}}\pseudo{(W')}}\lrp{\sqrt{p_k}W'\beta_k}} \\
&\leq \abs{\ptilde_k}\opnorm{\pseudo{(\Wtilde')}}\opnorm{\betatilde_k - \sqrt{p_k}W'\beta_k} \\
&\quad + \lrp{\abs{\ptilde_k - \frac{1}{\sqrt{p_k}}}\opnorm{\pseudo{(\Wtilde')}} + \frac{1}{\sqrt{p_k}}\opnorm{\pseudo{(\Wtilde')} - \pseudo{(W')}}}\cdot \enorm{\sqrt{p_k}W'\beta_k} \\
&\lsim \lrp{\eps + \frac{1}{\sqrt{\pmin}}}\sigma_K^{1/2} \frac{\eps}{\pmin} + \lrp{\eps \sigma_K^{1/2} + \frac{\eps \sigma_K^2}{\opnorm{M_3} \sqrt{\pmin}}}\cdot 1 \\
&\lsim \eps \lrp{\frac{\sigma_K^{1/2}}{\pmin^{3/2}} + \sigma_K^{1/2}  + \frac{\sigma_K^2}{\opnorm{M_3}\pmin^{1/2}}}\\
&\lsim \eps \lrp{\frac{\sigma_K^{1/2}}{\pmin^{3/2}} + \frac{\sigma_K^2}{\opnorm{M_3}\pmin^{1/2}}}
\end{align}
where we used that $\set{\sqrt{p_k}W'\beta_k}_{k \in [K]}$ are orthonormal, the bounds on $\pseudo{\Wtilde}$ and $\pseudo{\Wtilde} - \pseudo{W}$ from the $M_2$ analysis, and the tensor decomposition recovery bounds in \eqref{eq:pf-tensor-robust}. 

Next we bound the estimation error $\phat_k - p_k$. Let $a := \sqrt{\phat_k}$ and $b := \sqrt{p_k}$. Then \eqref{eq:pf-tensor-robust} is of the following form 
\begin{align}
& \abs{\frac{1}{a} - \frac{1}{b}} \lsim \eps \implies \abs{b-a} \lsim ab \eps, 
\end{align}
which also implies that $a \lsim\frac{b}{1-\eps b}$ when $\eps \lsim 1/b$.
Note that 
\begin{align}
\abs{\phat_k - p_k} &= \abs{a^2 - b^2} \\
&= \abs{a-b}(a+b) \\
&\leq \eps b^3 \lrp{1 + \frac{1}{1-\eps b}}\frac{1}{1-\eps b}
\end{align}
Since \eqref{eq:pf-eps-tensor} implies that $\eps \lsim 1/K^2 \leq 1/4$ when $K \geq 2$, we can bound $\frac{1}{1-\eps b}$ by a constant, and so we have that 
\begin{align}
\abs{\phat_k - p_k} &\lsim \eps p_k^{3/2}.
\end{align}
\end{proof}

\subsection{Comparison to \texorpdfstring{\citet{bakshi2023tensor}}{Bakshi et al.\ (2023b)}}\label{sec:bakshi-comparison}
In this section, we further elaborate on the relationship between the present work and \cite{bakshi2023tensor} that was discussed in Section \ref{sec:related-work} on related work.
\cite{bakshi2023tensor} estimates the tensor with $(i,j,l)$-th block component $\sum_{k \in [K]} p_k \lrp{C_k A_k^{i} B_k} \ot \lrp{C_k A_k^{j} B_k} \ot \lrp{C_k A_k^{l} B_k}$ using the tensor whose $(i,j,l)$-th component is 
\begin{align}\label{eq:bakshi-est}
\That(a, b, c) =\frac{1}{N} \sum_{i = 1}^N  y_{i, t_0+a+b+c+2} \ot u_{i, t_0+a+b+2} \ot y_{i, t_0+a+b+1} \ot u_{i, t_0+a+1} \ot y_{i, t_0+a} \ot u_{i, t_0},
\end{align}
with $t_0 = 0$.
For a parameter $s$ related to the observability and controllability of the systems, \cite{bakshi2023tensor} estimates the first $2s+1$ Markov parameters $\set{D_k, C_kA_k^i B_k \mid i = 0, ..., 2s, k \in [K]}$ and associated mixture weights $\set{p_k}_{k \in [K]}$ by decomposing the $mp(2s+1) \by mp(2s+1) \by mp(2s+1)$ tensor $\That$ constructed by flattening each of the component Markov parameter estimates into a vector. Notably their estimator only uses the first $6s$ samples from each trajectory, and estimation guarantees use concentration in the number of independent trajectory samples $N$ drawn from the mixture, and does not include any concentration in $T$, the length of each trajectory. 

Meanwhile our estimator is a whitened version of the following third order tensor, which more closely resembles the moment estimators used in standard mixtures of linear regression: 
\begin{align}
\Mtilde_3 = \frac{L}{NT} \sum_{i=1}^N \sum_{t \in\Jscr} \lrp{y_{i,t}^3 \ubar_{i,t}^{\ot 3} -  y_{i,t}^3 \Escr_k(\ubar_{i,t})},
\end{align}
but still estimates a third order tensor with the $(i,j,l)$-th entry as the mixture of the product of the $i$-th, $j$-th, and $l$-th Markov parameters of the systems.
While we assumed for simplicity that the scalar outputs $y_{i,t}$, we can easily extend our approach to multi-dimensional outputs $(p>1)$ by considering each dimension of the measured outputs separately and running the analogous moment estimator for each dimension. 
The form of our estimator allows us to use samples from the whole length of the trajectory $T$ as well as observations $N$, and to obtain estimation error upper bounds with concentration in both $N$ and $T$. 

Furthermore, \cite{bakshi2023tensor} relies on Jennrich's algorithm for tensor decomposition, which, while it can be applied to non-orthogonal tensors, does not come with explicit robustness guarantees. Rather, it is only shown that the error is polynomial in various dimensional parameters. In contrast, we use the tensor power iteration for tensor decomposition of the orthonormalized third order tensors, which comes with explicit robustness guarantees \cite{anandkumar2014tensor}, and is also practically efficient with fast convergence rates. 
\subsection{Proof of Proposition \ref{prop:MLDS-result}}\label{sec:MLDS-proof}
\begin{proof}
We apply Theorem \ref{thm:perturbed-MLR} for mixtures of linear regression with the mapping detailed in Section \ref{sec:analysis}. We bound the norm of the regression coefficients: 
\begin{align}
\enorm{\beta_i}^2 &= \sigma_u^2 \enorm{\gL_i}^2 = \sigma_u^2 \sum_{t=0}^{L-1} \enorm{g_i(t)}^2 \leq  \sigma_u^2 \Gamma_{\max} =:b^2.
\end{align}
Next, we note that $\eta_{i,t} = \brk{f_i\uL, \wbar_{it}}$ is subgaussian with variance proxy $\sigma_w^2 \opnorm{f_i\uL}^2 \leq \sigma_w^2 \Gamma_{\max}$.
Finally, the error term $\xi_{i,t}$ due to the truncated impulse response is also subgaussian with variance proxy 
\begin{align}
\sum_{j=L+1}^t (\sigma_u^2 + \sigma_w^2) \enorm{g(k)}^2 & \leq (\sigma_u^2 + \sigma_w^2) \lrp{C_\rho \rho^L \sum_{k=1}^{\infty} \rho^k}^2 \\
&\leq (\sigma_u^2 + \sigma_w^2)C_\rho \cdot \rho^L \cdot \frac{1}{1-\rho},
\end{align}
where we used the exponential decay rate bound $\rho$ as defined in the model assumptions. 
Then for any $R>0$, when 
\begin{align}\label{eq:pf-L-xi-cond}
L \geq \ln\lrp{\frac{C_{\rho}\rho (\sigma_u^2 + \sigma_w^2)}{R^2(1-\rho)}}/\ln\lrp{1/\rho},
\end{align}
we have $\sigma_\xi^2 \leq R^2$. 
Finally, plugging these components into Theorem \ref{thm:perturbed-MLR} with ambient dimension $d = Lm$, effective sample size $NT/L$, and 
\begin{align}
\sigma_y^2 = b^2 + \sigma_{\eta}^2 = (\sigma_u^2+\sigma_w^2) \Gamma_{\max}
\end{align}
we have that for $\eps > 0, \delta \in (0,1)$, when the following conditions hold: 
\begin{align}
\frac{N_2 T}{L} &\gsim \frac{\sigma_y^4\opnorm{M_3}^2}{\eps^2\pmin^2} \lrp{\sigma_K^5  \ln^4 \lrp{\frac{N_2 T\cdot 9^{Lm}}{\delta L}}\ln\lrp{\frac{9^{Lm}}{\delta}} +  \frac{\delta}{9^{Lm}}\cdot \frac{\sigma_y^4\opnorm{M_3}^2}{\eps^2\pmin^2}}, \\
\frac{N_3 T}{L} &\gsim \frac{ \sigma_y^6}{\eps^2\pmin^2\sigma_K^3}\lrp{ \ln^6 \lrp{\frac{33^K\cdot N_3T}{\delta L}}\ln\lrp{\frac{33^K}{\delta}}+\frac{\delta}{33^K}\cdot \frac{ \sigma_y^6}{\eps^2\pmin^2\sigma_K^3}}, \\
\sigma_\xi &\lsim \frac{\eps\pmin \sigma_K^3}{\sigma_y} \lrp{\frac{1}{\opnorm{M_3} \ln \lrp{\frac{9^{Lm} N_2}{\delta}}\ln\lrp{\frac{N_2}{\delta}}} \land \frac{1}{\sigma_y \ln^{3/2} \lrp{\frac{33^K N_3}{\delta}}\ln^{3/2} \lrp{\frac{N_3}{\delta}}}} \label{eq:pf-xi-cond}\\
\eps &\lsim  \frac{\sigma_y^3}{\sigma_K^{3/2} \pmin},
\end{align}
the error bounds in \eqref{eq:MLDS-err-bnd} hold with probability at least $1-\delta$, for $\delta > 0$. As the final step, we substitute in the right hand side of \eqref{eq:pf-xi-cond} as $R^2$ in \eqref{eq:pf-L-xi-cond} to obtain the condition on $L$ given the spectral radius bound $\rho$ on all components of the mixture, for the perturbations $\xi_{it}$ to be sufficiently bounded: 
\begin{align}
L \ln(1/\rho) \geq \ln \lrp{\frac{\sigma_y^4 \Gamma_{\max} C_{\rho} \rho}{\eps^2 \pmin^2 \sigma_K^3 (1-\rho)} \lrb{\Gamma_{\max}^{3/2} \ln^2\lrp{\frac{9^{Lm} N_2}{\delta}} + \sigma_y \ln^3 \lrp{\frac{33^K N_3}{\delta}}}},
\end{align}
where we used that $\opnorm{M_3} \leq \max_{k \in [K]} \enorm{\gscrL_k} \Gamma_{\max}^{3/2}$.
\end{proof}

\subsection{Lemmas for Estimating \texorpdfstring{$M_2$ and $W$}{M2 and W}}
In this section we present the statements and proofs of Lemma \ref{lem:M2-perturb}, which controls the effect of the perturbation $\xi$ in the estimation of the second order tensor $\Mhat_2$, Proposition \ref{prop:M2-conc-1}, which concentrates the estimate $\Mhat_2$ around its mean $M_2$, and Corollary \ref{cor:M2-conc} which provides sample complexity bounds for estimating $M_2$ by $\Mhat_2$. Lemma \ref{lem:single-covariance-tail} is an auxiliary lemma used to derive the concentration result. 
\begin{lemma} \label{lem:M2-perturb}
Let $\Mtilde_2$ and $\Mhat_2$ be as given in \eqref{eq:M23tilde-MLR} and \eqref{eq:M23hat-MLR}, respectively. For any $\delta \in (0,1)$, 
\begin{align}
\opnorm{\Mtilde_2 - \Mhat_2} \lsim \sigma_\xi (\sigma_y + \sigma_\xi)\ln\lrp{\frac{N}{\delta}}\ln \lrp{\frac{9^d N}{\delta}} 
\end{align}
with probability at least $1-\delta$. 
\end{lemma}
\begin{proof}
\begin{align}
\opnorm{\Mtilde_2 - \Mhat_2}&= \frac{1}{2N} \opnorm{\sum_{i=1}^N (\ytilde_i^2 - y_i^2 )\lrp{x_i\ot x_i - I_d}}\\
&\leq  \frac{1}{2N} \sum_{i=1}^N \abs{2\xi_i y_i + \xi_i^2}\opnorm{x_i\ot x_i - I_d} \\
&\leq\frac{1}{2N} \sum_{i=1}^N \lrp{4\sigma_y \sigma_\xi + 2\sigma_\xi^2} \ln \lrp{\frac{2}{\delta}} \lrp{2C\ln \lrp{\frac{2\cdot 9^d}{\delta}}} \\
&\lsim (\sigma_y \sigma_\xi + \sigma_\xi^2)\ln\lrp{\frac{2}{\delta}}\lrp{d + \ln \lrp{\frac{2}{\delta}}}
\end{align}
with probability at least $1-3N\delta$, using tail inequalities for subgaussian random variables $y_i$ and $\xi_i$ and a standard covariance concentration inequality Lemma \ref{lem:single-covariance-tail} with a union bound over all $i \in [N]$.
\end{proof}

\begin{proposition}\label{prop:M2-conc-1}
For any $\eps>0, t>1$,
\begin{align}\label{eq:intermed2-M2}
&\Pr{\opnorm{\Mhat_2 - M_2} \gtrsim  \eps+ \sigma_y^2 t^2\exp(-t^2/4)}\\
 & \quad \quad \leq 9^d \lrp{4 N \exp (-t^2/2) + 2 \exp \lrp{-\frac{N\eps^2}{8\sigma_y^4t^8 +(4/3) \sigma_y^2 t^4 \eps}}},
\end{align}
where $\sigma_y^2 = b^2 + \sigma_{\eta}^2$.
\end{proposition}
\begin{proof}
Corollary 4.2.13 of \cite{vershynin2018high} there exists a $1/4$-covering $\Cscr$ of $\Sscr^{d-1}$ in the Euclidean norm such that $\abs{\Cscr} \leq 9^{d}$. We bound the operator norm of the difference $\Mhat_2 - M_2$ over the $1/4$-covering (e.g., by Exercise 4.4.3 of \cite{vershynin2018high})
\begin{align}\label{eq:covering-M2}
\begin{split}
\opnorm{\Mhat_2 - \E[\Mhat_2]} &= \sup_{v \in \Sscr^{d-1}} \abs{\lrp{\Mhat_2 - \E[\Mhat_2]}(v,v)} \\
&\leq 4 \sup_{v \in \Cscr} \abs{\lrp{\Mhat_2 - \E[\Mhat_2]}(v,v)}.
\end{split}
\end{align}
Next,
\begin{align}
\Mhat_2(v,v) = \frac{1}{N} \sum_{i=1}^N \frac{1}{2}y_i^2 \lrp{\brk{v, x_i}^2-1} = \frac{1}{N} \sum_{i=1}^N Y_i
\end{align}
where we define $Y_i:= y_i^2 \lrp{\brk{v,x_i}^2 -1}/2$, and $\E[Y_i] = M_2(v,v) = \sum_{k=1}^K p_k \brk{v, \beta_k}^2$. 

Note that $y_i$ is a subgaussian random variable with variance proxy $\sigma_y^2:= b^2 + \sigma_\eta^2$ where recall $b = \max_{k \in [K]} \enorm{\beta_k}$. Define $w_i:= \brk{v, x_i}$, which is subgaussian with variance proxy $\sigma_w^2 := \sigma_x^2 = 1$. 

Fix $t>1$, let $t_y:= \sigma_y t$ and $t_w:= \sigma_w t=t$, and define the events $\Escr_{i,y}:= \set{\abs{y_i} \leq t_y}$, $\Escr_{i,w} := \set{\abs{w_i} \leq t_w}$, and $\Escr_i := \Escr_{i,y} \cap \Escr_{i,w}$. By standard subgaussian tail bounds, we have that the probability of $\Escr_{i,w}$ and of $\Escr_{i,y}$ are each upper bounded by $2\exp(-t^2/2)$, so that $\Pr{\Escr_i} \leq 4 \exp(-t^2/2)$. Finally, define $Z_i := Y_i \ind{\Escr_i}$. 

By the triangle inequality,
\begin{align}\label{eq:alt-triangle-M2}
\abs{\frac{1}{N} \sum_{i=1}^N \lrp{Y_i - \E[Y_i]}} &\leq 
\abs{\frac{1}{N} \sum_{i=1}^N \lrp{Y_i - Z_i}} + \abs{\frac{1}{N} \sum_{i=1}^N \lrp{Z_i - \E[Z_i]}} + \abs{\frac{1}{N} \sum_{i=1}^N \lrp{\E[Z_i] - \E[Y_i]}}
\end{align}
We bound each of the three summands on the right hand side of \eqref{eq:alt-triangle-M2} separately.

First, by construction of $Z_i$, we have that 
\begin{align}
\abs{\frac{1}{N} \sum_{i=1}^N (Y_i - Z_i)}&\leq \frac{1}{N}\sum_{i=1}^N \abs{Y_i}\ind{\Escr_i}.
\end{align}
This expression is nonzero with probability at most $\Pr{\cup_{i \in [N]} \Escr_i} \leq N\Pr{\Escr_i} \leq 4N\exp(-t^2/2)$ by a union bound.

Next, note that 
\begin{align}
\abs{Z_i} &= \abs{\frac{1}{6} y_i^3\lrp{w_i^3 - 3w_i \enorm{Wa}^2}} \ind{y_i \leq t_y} \ind{w_i \leq t_w} \\
&\leq \frac{1}{6} t_y^3 \lrp{t_w^3 + 3t_w \opnorm{W}^2}.
\end{align}
With $\sigma_w = \opnorm{W}$,
and assuming $t>1$, we have $\abs{Z_i} \leq \frac{1}{6} \sigma_y^3 \sigma_w^3 t^3(t^3 + 3t) \leq \frac{2}{3} \sigma_y^3 \sigma_w^3 t^6=:B(t)/2$. 

Then $\abs{Z_i - \E[Z_i]} \leq B(t)$, and we can crudely bound $\sum_{i=1}^N \E[(Z_i - \E[Z_i])^2]$ by $NB(t)^2$.
Then, noting the independence of $Z_i$ across $i \in [N]$, we apply the Bernstein inequality for independent bounded random variables \citep[Theorem 2.8.4]{vershynin2018high}, to get that for any $\eps>0$, 
\begin{align}
\Pr{\abs{\frac{1}{N}\sum_{i=1}^N (Z_i - \E[Z_i])} > \eps} \leq 2\exp \lrp{-\frac{N\eps^2}{2B(t)^2 + (2/3)B(t)\eps}}.
\end{align}
Plugging in $B(t)= (4/3) \sigma_y^3 \sigma_w^3 t^6$ gives us
\begin{align}
\Pr{\abs{\frac{1}{N}\sum_{i=1}^N (Z_i - \E[Z_i])} > \eps} \leq 2\exp \lrp{-\frac{9N\eps^2}{32\sigma_y^6 \sigma_w^6 t^{12} + 8 \sigma_y^3 \sigma_w^3 t^6 \eps}}.
\end{align}

Finally, we bound the difference in means of the $Y_i$ and its truncated version $Z_i$, to get
\begin{align}
\abs{\E[Z_i] - \E[Y_i]} \lsim \sigma_y^2 (\sigma_w \lor 1) t^2 \exp(-t^2/4).
\end{align}
The proof is more involved so the step is presented in Lemma \ref{lem:truncation-expectation-M3}.

Altogether, we get that for all $\eps \geq 0, t \geq 1$, with $\sigma_w = 1$, 
\begin{align}
&\Pr{\abs{\frac{1}{N} \sum_{i=1}^N Y_i - \E[Y_i]} \geq \eps+C\sigma_y^2 t^2\exp(-t^2/4)}\\
&\quad \leq 4 N \exp (-t^2/2) +2 \exp \lrp{-\frac{N\eps^2}{8\sigma_y^4t^8 +(4/3) \sigma_y^2 t^4 \eps}}.
\end{align}
Recalling that $\Mhat_2(v,v) = \frac{1}{N} \sum_{i=1}^N Y_i$ for a fixed $v \in \Cscr$, we now apply a union bound over all $v \in \Cscr$ and use \eqref{eq:covering-M2} to conclude that
\begin{align}
& \Pr{\opnorm{\Mhat_2 - M_2} \geq 4 \lrp{\eps+\sigma_y^2 t^2\exp(-t^2/4)}}\\
 & \quad \quad \leq 9^d \lrp{4 N \exp (-t^2/2) + 2 \exp \lrp{-\frac{N\eps^2}{8\sigma_y^4t^8 +(4/3) \sigma_y^2 t^4 \eps}}}.
\end{align}
\end{proof}

\begin{corollary}\label{cor:M2-conc}
Given $\eps>0$ and $\delta \in (0,1)$ when
\begin{align}\label{eq:N2-condition}
N \gsim \max \set*{\frac{\sigma_y^4}{\eps^2} \ln^4 \lrp{\frac{4N\cdot 9^d}{\delta}}\ln\lrp{\frac{2\cdot 9^d}{\delta}}, \frac{\delta}{9^d}\lrp{\frac{\sigma_y^8}{\eps^4} \lor \exp\lrp{\frac{\eps^{1/2}}{\sigma_y}}}} 
\end{align}
then 
\begin{align}
\Pr{\opnorm{\Mhat_2 - M_2} \geq \eps} \leq \delta.
\end{align}
When $\eps < \sigma_y^2/1.51$, it suffices to have
\begin{align}\label{eq:N2-condition-simplified}
N \gsim \max \set*{\frac{\sigma_y^4}{\eps^2} \ln^4 \lrp{\frac{4N\cdot 9^d}{\delta}}\ln\lrp{\frac{2\cdot 9^d}{\delta}}, \frac{\delta}{9^d}\cdot \frac{\sigma_y^8}{\eps^4}}.
\end{align}
\end{corollary}

\begin{proof}
Starting from the result of Proposition \ref{prop:M2-conc-1}, %
let $t$ be such that 
\begin{align}
&\delta = 9^d\cdot 4 N \exp (-t^2/2) \\
\implies & t^2 =2 \ln \lrp{\frac{9^d\cdot 4N}{\delta}}.
\end{align}
Note that there exists a universal constant $C_1$ such that
\begin{align}
t^2 \exp(-t^2/4) \leq C_1 \exp(-t^2/8).
\end{align}
To ensure a small error we set an implicit condition on $N$ by setting 
\begin{align}
&\eps \gsim \sigma^2  \exp(-t^2/8) \gsim \sigma_y^2 t^2 \exp(-t^2/4)\\
\iff &8\ln\lrp{\frac{\sigma_y^2}{\eps}} \leq t^2  = 2\ln\lrp{\frac{9^d\cdot 4N}{\delta}}\\
\iff & N \gsim \frac{\delta \sigma_y^{8}}{9^d \eps^4}. \label{eq:intermed-M2-N1}
\end{align}

Finally, when $\eps\lsim \sigma_y^2 t^4$, the second term in the probability bound simplifies and we set 
\begin{align}
&\delta \geq 9^d\cdot 2\exp\lrp{-\frac{N\eps^2}{10\sigma_y^4 t^{8}}}\\
\iff &N \gsim \frac{\sigma_y^4}{\eps^2}t^{8} \ln\lrp{\frac{2\cdot 9^d}{\delta}} \\
\iff  & N \gtrsim \frac{\sigma_y^4}{\eps^2} \ln^4 \lrp{\frac{4N\cdot 9^d}{\delta}}\ln\lrp{\frac{2\cdot 9^d}{\delta}}. \label{eq:intermed-M2-N2}
\end{align}
The condition for the simplification is implied by 
\begin{align}
&\eps \lsim \sigma_y^2 \ln^2 \lrp{\frac{9^d \cdot 4N}{\delta}} \\
\iff& N \gsim \frac{\delta}{9^d}\exp\lrp{\frac{\eps^{1/2}}{\sigma_y}} \label{eq:intermed-M2-N3}
\end{align}
which is easily satisfied in general. 
In fact, note that $x^2 > \exp(x^{-1/2})$ whenever $x \geq 1.51$, so setting $x = \sigma_y^2/\eps$, and since clearly $\ln(2\cdot 9^d/\delta) > \delta/9^d$ for $\delta \in (0,1)$, we have that when $\eps < \sigma_y^2/1.51$, the condition in \eqref{eq:intermed-M2-N3} is redundant in view of \eqref{eq:intermed-M2-N2}.

In general, under the combined conditions \eqref{eq:intermed-M2-N1}, \eqref{eq:intermed-M2-N2}, and \eqref{eq:intermed-M2-N3}, we get the claimed error bound.
\end{proof}

\begin{lemma}\label{lem:truncation-expectation-M2}
Let $t = t_w/\sigma_w = t_y/\sigma_y \geq 1$.
\begin{align}
\abs{\E[Z_i] - \E[Y_i]} \lsim \sigma_y^2 (\sigma_w \lor 1) t^2 \exp(-t^2/4).
\end{align}
\end{lemma} 
\begin{proof}
We upper bound $\abs{\E[Z_i] - \E[Y_i]}= \abs{\E[Y_i \ind{\Escr_i^c}]}$.
\begin{align}
2\abs{\E[Y_i \ind{\Escr_i^c}]} &\leq 2\Exp{\abs{Y_i} \ind{\Escr_i^c}} \notag\\
&\leq \Exp{\abs{y_i^2w_i^2} \ind{\Escr_i^c}} + \Exp{\abs{y_i^2} \ind{\Escr_i^c}}.\label{eq:intermed1-M2}
\end{align}

Note that we can decompose the indicator of event $\Escr_i^c$ as $\ind{\Escr_i^c} = \ind{\Escr_{i,y}^c \cap \Escr_{i,w}^c} +
 \ind{\Escr_{i,y} \cap \Escr_{i,w}^c} + 
 \ind{\Escr_{i,y}^c \cap \Escr_{i,w}}$.
Focusing on the first term of \eqref{eq:intermed1-M2}, we have 
\begin{align}
\Exp{\abs{y_i^2w_i^2} \ind{\Escr_i^c}} &\leq\underbrace{\Exp{\abs{y_i^2w_i^2} \ind{\Escr_{i,y}^c \cap \Escr_{i,w}^c}}}_{(A)}  + \underbrace{\Exp{\abs{y_i^2w_i^2} \ind{\Escr_{i,y} \cap \Escr_{i,w}^c}}}_{(B)}  + \underbrace{\Exp{\abs{y_i^2w_i^2} \ind{\Escr_{i,y}^c \cap \Escr_{i,w}}}}_{(C)}.
\end{align}

By the Cauchy-Schwarz inequality, 
\begin{align}
(A) &\leq \sqrt{ \Exp{y_i^4\ind{\Escr_{i,y}^c \cap \Escr_{i,w}^c}} \Exp{w_i^4\ind{\Escr_{i,y}^c \cap \Escr_{i,w}^c}} } \\
&\leq \sqrt{ \Exp{y_i^4\ind{\Escr_{i,y}^c}} \Exp{w_i^4 \ind{\Escr_{i,w}^c}}}. 
\end{align}
Corollary \ref{cor:subg-trunc-moment} implies
\begin{align}
\Exp{y_i^4 \ind{\abs{y_i}> t_y}} &\leq 8\lrp{2\sigma_y^4 + \sigma_y^2t_y^{2}}  \exp\lrp{-\frac{t_y^{1/2}}{2\sigma_y^2}} 
\end{align}
and similarly for $w_i$ so that 
\begin{align}
(A) \leq 8 \lrp{2\sigma_w^4 + \sigma_w^2t_w^{2}}^{1/2} \lrp{2\sigma_y^4 + \sigma_y^2t_y^{2}}^{1/2}
 \exp\lrp{-\frac{t_y^{2}}{4\sigma_y^2} -\frac{t_w^{2}}{4\sigma_w^2}}.
\end{align}
With $t = t_w/\sigma_w = t_y/\sigma_y \geq 1$, this simplifies to $(A) \lsim \sigma_w^2 \sigma_y^2 \exp{-t^2/2}$.
Similarly, by the Cauchy-Schwarz inequality, we have 
\begin{align}
(B) &= \Exp{\abs{y_i^2 w_i^2} \ind{\set{\abs{y_i} \leq t_y} \cap \set{\abs{w_i}> t_w}}}  \\
&\leq \sqrt{\Exp{y_i^4}  \Exp{w_i^4 \ind{\abs{w_i} > t_w}}} \\
&\leq 72\sigma_y^2 \cdot \sqrt{8} \lrp{2\sigma_w^4 + \sigma_w^2t_w^{2}}^{1/2} \exp\lrp{-\frac{t_w^{2}}{4\sigma_w^2}} \\
&\lsim \sigma_w^2\sigma_y^2 t \exp(-t^2/4)
\end{align}
where the third line follows from a bound on the moment of the subgaussian variable $y_i$ (c.f.  \cite[Proposition 2.5.2]{vershynin2018high}) and from Corollary \ref{cor:subg-trunc-moment}.

Finally, 
\begin{align}
(C) &= \Exp{\abs{y_i^2 w_i^2} \ind{\abs{y_i}> t_y}\cap\ind{\abs{w_i} \leq t_w}}  \\
&\leq \sqrt{\Exp{y_i^4\ind{\abs{y_i}> t_y}}  \Exp{w_i^4}} \\
&\leq   2\sqrt{2}\lrp{2\sigma_y^4 + \sigma_y^2t_y^{2}}^{1/2}  \exp\lrp{-\frac{t_y^{2}}{4\sigma_y^2}}\cdot (72\sigma_w^2) \\
&\lsim \sigma_w^2 \sigma_y^2 t \exp(-t^2/4).
\end{align}
Altogether, we have 
\begin{align}
\Exp{\abs{y_i^2w_i^2} \ind{\Escr_i^c}} \lsim \sigma_w^2 \sigma_y^2(t^2 + 2t)\exp(-t^2/4).
\end{align}

Following a similar procedure for the second term, we have 
\begin{align}
\Exp{\abs{y_i^2} \ind{\Escr_i^c}} &\leq \underbrace{\Exp{\abs{y_i^2} \ind{\Escr_{i,y}^c \cap \Escr_{i,w}^c}}}_{(A)}  + \underbrace{\Exp{\abs{y_i^2} \ind{\Escr_{i,y} \cap \Escr_{i,w}^c}}}_{(B)}  + \underbrace{\Exp{\abs{y_i^2} \ind{\Escr_{i,y}^c \cap \Escr_{i,w}}}}_{(C)}
\end{align}
with 
\begin{align}
(A) &\leq \sqrt{\Exp{y_i^4 \ind{\abs{y_i}>t_y}}\Exp{\ind{\abs{w_i}>t_w}}} \\
&\leq \sqrt{8}\lrp{2\sigma_y^4 + \sigma_y^2t_y^{1/2}}^{1/2} \exp\lrp{-\frac{t_y^{1/2}}{4\sigma_y^2}}
\cdot \sqrt{2}\exp(-t_w^2/4\sigma_w^2), \\
&\lsim \sigma_y^2 t \exp(-t^2/2)\\
(B) &\leq\sqrt{\Exp{y_i^4} \Exp{\ind{\abs{w_i}>t_w}}} \\
&\leq 72\sigma_y^2\sqrt{2\exp(-t_w^2/\sigma_w^2)} \lsim \sigma_y^2 \exp(-t^2/2), \\
(C) &\leq \sqrt{\Exp{y_i^4 \ind{\abs{y_i}>t_y}}\Exp{1}} \\
&\leq \sqrt{8\lrp{2\sigma_y^4 + \sigma_y^2t_y^{1/2}}  \exp\lrp{-\frac{t_y^{1/2}}{2\sigma_y^2}}}  \lsim \sigma_y^2 \exp(-t^2/4)
\end{align}
where we also used the subgaussian tail bound $\Pr{\abs{w_i}>t_w} \leq 2\exp(-t_w^2/(2\sigma_w^2))$.
Plugging in the inequalities we get 
\begin{align}
\Exp{\abs{y_i^2} \ind{\Escr_i^c}} &\lsim \sigma_y^2 t \exp(-t^2/4).
\end{align}

In all, we have
\begin{align}\label{eq:bigfat1-M2}
\begin{split}
\abs{\E[Z_i] - \E[Y_i]} & \lsim \sigma_w^2 \sigma_y^2 (t^2+2t) \exp(-t^2/4) + \sigma_y^2 t \exp(-t^2/4) \\
&\lsim \sigma_y^2 (\sigma_w \lor 1) t^2 \exp(-t^2/4).
\end{split}
\end{align}
\end{proof}

\begin{lemma}[Tail bound for single sample covariance] \label{lem:single-covariance-tail}
Let $X$ be an isotropic subgaussian random vector in $\R^d$. Then for any $\delta \in (0,1)$, 
\begin{align}
\Pr{\opnorm{X\ot X - I_d} \geq C \ln \lrp{\frac{2\cdot 9^d}{\delta}}} \leq \delta.
\end{align}
\end{lemma}
\begin{proof}
By Exercise 4.4.3 of \cite{vershynin2018high}, there exists a cover $\Cscr$ of $\Sscr^{d-1}$ such that $\abs{\Cscr} \leq 9^d$ and 
\begin{align}\label{eq:intermed-sample-cov}
\opnorm{X\ot X - I_d} \leq 2\sup_{a \in \Cscr} \abs{\brk{X, a}^2 - 1}.
\end{align}
For every $a \in \Cscr$, $\brk{X, a}$ is sub-gaussian with variance proxy $1$, so $\brk{X,a}^2 -1$ is zero-mean sub-exponential with parameter $K$ where $K$ is a universal constant. Applying a sub-exponential tailbound on this quantity (e.g., Proposition 2.7.1(a) in \cite{vershynin2018high}), we get
\begin{align}
&\Pr{\abs{\brk{x,a}^2 - 1} \geq C \ln \lrp{\frac{2}{\delta}}} \leq \delta 
\end{align}
Taking a union bound over all $a \in \Cscr$ and plugging this result back into \eqref{eq:intermed-sample-cov}, we obtain the result.
\end{proof}
\subsection{Lemmas for Estimating \texorpdfstring{$M_3$}{M3}}
In this section we present the statements and proofs of Lemma \ref{lem:M3-perturb}, which controls the effect of the perturbation $\xi$ in the estimation of the third order tensor $\Mhat_3$, Proposition \ref{prop:tensor-conc1}, which concentrates the estimate $\Mhat_3^V$ around its mean $M_3^V$, for any whitening matrix $V$, and Corollary \ref{cor:tensor-conc} which provides sample complexity bounds for estimating $M_3$ by $\Mhat_3$. 
Lemma \ref{lem:truncation-expectation-M3} is an auxiliary lemma used to derive the concentration result. 
\begin{lemma}\label{lem:M3-perturb}
Let $\Mtilde_3$ and $\Mhat_3$ be as defined in \eqref{eq:M23tilde-MLR} and \eqref{eq:M23hat-MLR}, respectively, with $\sigma_\xi \leq 1$. For any $d\by K$ matrix $V$, for any $\eps > 0$ and $\delta \in (0,1)$, with probability at least $1-\delta$, 
\begin{align}
\opnorm{\Mtilde_3^V - \Mhat_3^V} \lsim \sigma_\xi \sigma_y^2 \opnorm{V}^{3} \ln^{3/2}\lrp{\frac{N}{\delta}} \ln^{3/2}\lrp{\frac{N\cdot 33^K}{\delta}}.
\end{align}
\end{lemma}
\begin{proof}
We have that 
\begin{align}
\opnorm{\Mtilde_3^V - \Mhat_3^V} &\leq \frac{1}{6N_2} \sum_{i=1}^{N_2} \abs{3y_i^2 \xi_i + 3y_i \xi_i^2 + \xi_i^3}\opnorm{(V'x_i)^{\ot 3} - \Escr(x_i)(V,V,V)}.
\end{align}
Using subgaussian tail inequalities for $y_i$ and $\xi_i$ with a union bound over all $i \in [N]$ we have that with probability at least $1-2N\delta$, 
\begin{align}\label{eq:pf:intermed-1}
\abs{3y_i^2 \xi_i + 3y_i \xi_i^2 + \xi_i^3} \leq 3\cdot 2^{3/2}  \lrp{\sigma_y^2 \sigma_\xi + \sigma_y\sigma_\xi^2 + \sigma_\xi^3} \ln^{3/2}\lrp{\frac{2}{\delta}}
\lsim \sigma_\xi \sigma_y^2 (\sigma_\xi^2 \lor 1) \ln^{3/2}\lrp{\frac{2}{\delta}} , 
\end{align}
where we recall that $\sigma_y\geq \sigma_x \geq 1$.
Next, fix an $i \in [N]$ and temporarily let $X = V'x_i$, which is a subgaussian random vector with variance proxy at most $\opnorm{V}^2$. Let $\Cscr \subset \Sscr^{K-1}$ be a (1/16)-cover for $\Sscr^{K-1}$ of size at most $33^K$ such that  
\begin{align}
\opnorm{(V'x_i)^{\ot 3} - \Escr(x_i)(V, V, V)} \leq \sup_{u \in \Cscr}  16\abs{\brk{X, u}^3} + 3\abs{\brk{X, u}}\sum_{k=1}^d \brk{v_k, u}^2
\end{align}
Note we can rewrite $\sum_{k=1}^d \brk{v_k, u}^2 = \brk{u, V'V u} \leq \opnorm{V}^2$. Further, $\brk{X, u}$ is also subgaussian with variance proxy at most $\opnorm{V}^2$. Plugging in standard subgaussian concentration inequality where for any $\delta \in (0,1)$, 
\begin{align}
\Pr{\abs{\brk{X,u}} > \sqrt{2\opnorm{V}^2 \ln(2/\delta)}} \leq \delta,
\end{align}
and union bounding over all $i \in [N]$ and $u \in \Cscr$, we have that given $\delta \in (0,1)$, it holds with probability at least $1-\delta$ that for all $i \in [N]$, 
\begin{align}\label{eq:pf:intermed-2}
\abs{\opnorm{(V'x_i)^{\ot 3} - \Escr(x_i)(V, V, V)}} \lsim \opnorm{V}^3 \ln^{3/2}\lrp{\frac{N\cdot 33^K}{\delta}}.
\end{align}
Combining the two bounds \eqref{eq:pf:intermed-1} and \eqref{eq:pf:intermed-2}, with $\sigma_\xi \leq 1$, gives us
\begin{align}
\opnorm{\Mtilde_3^V - \Mhat_3^V} &\lsim \sigma_\xi \sigma_y^2 \opnorm{V}^3 \ln^{3/2}\lrp{\frac{N}{\delta}} \ln^{3/2} \lrp{\frac{N\cdot 33^K}{\delta}} 
\end{align}
with probability at least $1-\delta$.
\end{proof}

\begin{proposition}\label{prop:tensor-conc1}
Let $V$ be any $d\by K$ matrix. For any $t>1, \eps>0$, 
\begin{align}\label{eq:intermed2-M3}
& \Pr{\opnorm{\firstV - \Exp{\firstV}} \gsim \eps+C \opnorm{V}^3 \sigma_y^3 t^4 \exp (-t^2/4)}\\
 & \quad \leq 33^K \lrp{4 N \exp (-t^2/2) + 2 \exp \lrp{-\frac{9N\eps^2}{32\sigma_y^6 \opnorm{V}^6 t^{12} + 8\sigma_y^3\opnorm{W}^3 t^6 \eps}}},
\end{align}
where $\sigma_y^2:= b^2 +\sigma_\eta^2$.
\end{proposition}
\begin{proof}
Within this proof we write $W$ in place of $V$, but here it represents any arbitrary $d\by K$ matrix, not necessarily the whitening matrix. By Corollary 4.2.13 of \cite{vershynin2018high} with $\eps = 1/16$, there exists a $1/16$-covering $\Cscr$ of $\Sscr^{K-1}$ in the Euclidean norm such that $\abs{\Cscr} \leq 33^{K}$. By Lemma \ref{lem:covering}, 
\begin{align}\label{eq:covering}
\begin{split}
\opnorm{\first - \Exp{\first}} &= \sup_{a \in \Sscr^{K-1}} \abs{\lrp{\first - \E[\first]}(a, a, a)} \\
&\leq 16 \sup_{a \in \Cscr} \abs{\lrp{\first - \E[\first]}(a, a, a)}.
\end{split}
\end{align}
We will bound 
\begin{align}
\abs{\first - \E[\first](a, a, a)} = \lrb{\Exp{\brk{Wa, \beta}^3} - \frac{1}{6N}\sum_{i=1}^N y_i^3 \lrp{\brk{Wa, x_i}^3 - \Escr(x_i)(Wa, Wa, Wa)}}
\end{align}
for an arbitrary $a \in \Cscr$, then apply a union bound over $\Cscr$.

First, we simplify expressions by evaluating $\Escr(x_i)(Wa, Wa, Wa)$, which is a scalar:
\begin{align}
\Escr(x_i)(Wa, Wa, Wa) &= 3\sum_{j=1}^d \brk{Wa,x_i} \brk{Wa,e_j}^2  \\
&= 3\brk{Wa,x_i}\Tr\lrp{(Wa)(Wa)'\sum_{j=1}^d e_j e_j'}\\
&= 3\brk{Wa,x_i}\brk{Wa, Wa} = 3\brk{Wa,x_i}\enorm{Wa}^2. 
\end{align}
Thus we wish to show concentration of
\begin{align}
\first(a,a,a) = \sum_{i=1}^N\lrp{ \frac{1}{6N} y_i^3 \lrp{\brk{Wa, x_i}^3 - 3 \brk{Wa, x_i} \enorm{Wa}^2} } = \frac{1}{N} \sum_{i=1}^N Y_i
\end{align}
where we define $Y_i := \frac{1}{6} y_i^3\lrp{\brk{Wa, x_i}^3 - 3\brk{Wa, x_i} \enorm{Wa}^2}$ and $\E[Y_i] = \Exp{\brk{a,\beta}^3}$

Let $w_i:=\brk{Wa,x_i}$. We bound $y_i$ and $w_i$ and their powers with high probability by noting that both terms are subgaussian with variance proxies $\sigma_y^2:=b^2\sigma_x^2 + \sigma_\eta^2 = b^2 + \sigma_\eta^2$ and $\sigma_w^2:= \opnorm{W}^2\geq \enorm{Wa}^2\sigma_x^2$ (with $\sigma_x^2 = 1$), respectively. 

Fix $t>0$ and let $t_y := \sigma_y t$ and $t_w := \sigma_w t$. Define the events
\begin{align}
\Escr_{i,y} &:= \set{\abs{y_i} \leq t_y} \\
\Escr_{i,w} &:= \set{\abs{w_i} \leq t_w},\\
\Escr_i &:= \Escr_{i,y}\cap \Escr_{i,w}
\end{align}
and let $Z_i:= Y_i \ind{\Escr_i}$ be a truncated version of $Y_i$. Using the subgaussian tail bounds
\begin{align}
\Pr{\abs{y_i} \geq t_y} &\leq 2\exp \lrp{-t_y^2/(2\sigma_y^2)} = 2\exp(-t^2/2)\\
\Pr{\abs{w_i} \geq t_w} &\leq 2\exp \lrp{-t_w^2/(2\sigma_w^2)} = 2\exp(-t^2/2), 
\end{align}
we have that $\Pr{\Escr_i} \leq \Pr{\Escr_{i,y}} + \Pr{\Escr_{i,w}} \leq 4\exp(-t^2/2)$.
We follow (parts of) \cite{yi2016solving} and \cite{zhong2016mixed} and use a triangle inequality to show concentration of $(1/N)\sum_i Y_i$ via concentration of $(1/N) \sum_i Z_i$ which has bounded summands.

By the triangle inequality,
\begin{align}\label{eq:alt-triangle}
\abs{\frac{1}{N} \sum_{i=1}^N \lrp{Y_i - \E[Y_i]}} &\leq 
\abs{\frac{1}{N} \sum_{i=1}^N \lrp{Y_i - Z_i}} + \abs{\frac{1}{N} \sum_{i=1}^N \lrp{Z_i - \E[Z_i]}} + \abs{\frac{1}{N} \sum_{i=1}^N \lrp{\E[Z_i] - \E[Y_i]}}
\end{align}
We bound each of the three summands on the right hand side of \eqref{eq:alt-triangle} separately.

First, by construction of $Z_i$, we have that 
\begin{align}
\abs{\frac{1}{N} \sum_{i=1}^N (Y_i - Z_i)}&\leq \frac{1}{N}\sum_{i=1}^N \abs{Y_i}\ind{\Escr_i}.
\end{align}
This expression is nonzero with probability at most $\Pr{\cup_{i \in [N]} \Escr_i} \leq N\Pr{\Escr_i} \leq 4N\exp(-t^2/2)$ by a union bound. 

Next, by definition, $\abs{Z_i} = \abs{\frac{1}{2} y_i^2 (w_i^2 -1)} \ind{\abs{y_i} \leq t_y, \abs{w_i} \leq t_w} \leq \frac{1}{2} t_y^2(t_w^2+1)$. We bound $\abs{Z_i - \E[Z_i]} \leq t_y^2(t_w^2 + 1) =\sigma_y^2t^2(\sigma_w^2t^2+1) \leq 2\sigma_y^2t^4$ where we used that $\sigma_w = 1$ and $t>1$. By the Bernstein bound for independent bounded random variables \citep[Theorem 2.8.4]{vershynin2018high}, we have that for any $\eps>0$,
\begin{align}
\Pr{\abs{\frac{1}{N}\sum_{i=1}^N (Z_i - \E[Z_i])} > \eps} \leq 2 \exp \lrp{-\frac{N\eps^2}{8\sigma_y^4t^8 +(4/3) \sigma_y^2 t^4 \eps}}.
\end{align}

Finally, we bound the difference in means of the $Y_i$ and its truncated version $Z_i$, to get
\begin{align}
\abs{\E[Z_i] - \E[Y_i]} \lsim \sigma_y^2 (\sigma_w \lor 1) t^2 \exp(-t^2/4).
\end{align}
The proof is more involved so the step is presented in Lemma \ref{lem:truncation-expectation-M2}.

Altogether, we get that for all $\eps \geq 0, t \geq 1$, 
\begin{align}
&\Pr{\abs{\frac{1}{N} \sum_{i=1}^N Y_i - \E[Y_i]} \geq \eps+C \sigma_w^3\sigma_y^3 t^4 \exp \lrp{-\frac{t^2}{4}}}\\
&\quad \leq 4 N \exp (-t^2/2) + 2 \exp \lrp{-\frac{9N\eps^2}{32\sigma_y^6 \sigma_w^6 t^{12} + 8 \sigma_y^3 \sigma_w^3 t^6 \eps}}
\end{align}

Recalling that $\first(a,a,a) = \frac{1}{N} \sum_{i=1}^N Y_i$ for a fixed $a \in \Cscr$, we now apply a union bound over all $a \in \Cscr$ and use \eqref{eq:covering} to conclude that
\begin{align}
& \Pr{\opnorm{\first - \E[\first]} \geq 16 \lrp{\eps+C \sigma_w^3\sigma_y^3 t^4 \exp \lrp{-\frac{t^2}{4}}}}\\
 & \quad \quad \leq 33^K \lrp{4 N \exp (-t^2/2) + 2 \exp \lrp{-\frac{9N\eps^2}{32\sigma_y^6 \sigma_w^6 t^{12} + 8\sigma_y^3 \sigma_w^3 t^6 \eps}}}.
\end{align}
To get the final result we plug in $\sigma_w = \opnorm{W}$.
\end{proof}

\begin{corollary}\label{cor:tensor-conc}
For any matrix $V \in \R^{d\by K}$, any $\eps >0$ and $\delta<1$, when 
\begin{align}\label{eq:N3-condition}
N \gsim \max \set*{ \frac{ \sigma_y^6 \opnorm{V}^6}{\eps^2} \ln^6 \lrp{\frac{33^K\cdot 4N}{\delta}}\ln\lrp{\frac{2\cdot 33^K}{\delta}},\frac{\delta}{33^K}\lrp{\frac{\sigma_y^{12}\opnorm{V}^{12}}{\eps^4} 
\lor \exp\lrp{\frac{\eps^{1/3}}{\sigma_y \opnorm{V}}}}},
\end{align}
where $\sigma_y^2 = b^2 + \sigma_\eta^2$, 
then $\Pr{\opnorm{\firstV - \E[\firstV]} \gsim \eps} \leq \delta.$

When $\eps < \sigma_y^3\opnorm{V}^3/1.55$, it suffices to have 
\begin{align}\label{eq:N3-condition-simplified}
N \gsim \max \set*{ \frac{ \sigma_y^6 \opnorm{V}^6}{\eps^2} \ln^6 \lrp{\frac{33^K\cdot 4N}{\delta}}\ln\lrp{\frac{2\cdot 33^K}{\delta}},\frac{\delta}{33^K}\cdot \frac{\sigma_y^{12}\opnorm{V}^{12}}{\eps^4}}.
\end{align}
\end{corollary}

\begin{proof}
Starting from the result of Proposition \ref{prop:tensor-conc1}, %
set the truncation level $t$ such that
\begin{align}
&\delta = 33^K\cdot 4 N \exp (-t^2/2) \\
\implies & t^2 =2 \ln \lrp{\frac{33^K \cdot 4N}{\delta}}.
\end{align}
Note that there exists a universal constant $C_1$ such that
\begin{align}
t^4 \exp(-t^2/4) \leq C_1 \exp(-t^2/8).
\end{align}

For simplicity let $\sigma:= \opnorm{V}\sigma_y$.
To ensure a small error we set an implicit condition on $N$ by setting 
\begin{align}
&\eps \gsim \sigma^3  \exp(-t^2/8) \gsim \sigma^3 t^4 \exp(-t^2/4)\\
\iff &8\ln\lrp{\frac{\sigma^3}{\eps}} \leq t^2  = 2\ln\lrp{\frac{33^K\cdot 4N}{\delta}}\\
\iff & N \gsim \frac{\delta \sigma^{12}}{33^K \eps^4}. \label{eq:intermed-M3-N1}
\end{align}

Finally, when $\eps < 4\sigma^3 t^6$, the second term in the probability bound simplifies and we set 
\begin{align}
&\delta \geq 33^K\cdot 2\exp\lrp{-\frac{9N\eps^2}{64\sigma^6 t^{12}}}\\
\iff &N \gsim \frac{\sigma^6}{\eps^2}t^{12} \ln\lrp{\frac{2\cdot 33^K}{\delta}} \\
\iff  & N \gtrsim \frac{\sigma^6}{\eps^2} \ln^6 \lrp{\frac{4N\cdot 33^K}{\delta}}\ln\lrp{\frac{2\cdot 33^K}{\delta}}. \label{eq:intermed-M3-N2}
\end{align}
The condition for the simplification is implied by 
\begin{align}
&\eps \lsim \sigma^3 \ln^3 \lrp{\frac{33^K \cdot 4N}{\delta}} \\
\iff& N \gsim \frac{\delta}{33^K}\exp\lrp{\frac{\eps^{1/3}}{\sigma}} \label{eq:intermed-M3-N3}
\end{align}
which is easily satisfied in general. 
In fact, note that since $x^2 > \exp(x^{-1/3})$ whenever $x \geq 1.55$, let $x = \sigma_y^3 \opnorm{V}^3/\eps$, and note that $\ln(2\cdot 33^K/\delta) \geq \delta/33^K$ for $\delta \in (0,1)$. Then we have that when $\eps<\sigma_y^3 \opnorm{V}^3/1.55$, condition \eqref{eq:intermed-M3-N3} on $N$ is redundant in view of \eqref{eq:intermed-M3-N2}.

In general, under the combined conditions \eqref{eq:intermed-M3-N1}, \eqref{eq:intermed-M3-N2}, and \eqref{eq:intermed-M3-N3}, we get the claimed error bound.
\end{proof}

\begin{lemma}\label{lem:truncation-event}
For $i \in [N]$, let $Z_i = Y_i \ind{\Escr_i}$, with $(Y_i, Z_i)$ independent across $i$. Then 
\begin{align}
 \Pr{\abs{\frac{1}{N} \sum_{i=1}^N \lrp{Y_i - Z_i}} > 0} \leq N\Pr{\Escr_i}
\end{align}
\end{lemma}
\begin{proof}
By construction, we have 
\begin{align}
\abs{\frac{1}{N} \sum_{i=1}^N (Y_i - Z_i)}&\leq \frac{1}{N}\sum_{i=1}^N \abs{Y_i}\ind{\Escr_i}.
\end{align}
This expression is nonzero with probability at most $\Pr{\cup_{i \in [N]} \Escr_i} \leq 4N\exp(-t^2/2)$ by union bound. 
\end{proof}

\begin{lemma}\label{lem:truncation-expectation-M3}
Let $t = t_w/\sigma_w = t_y/\sigma_y \geq 1$. Then 
\begin{align}
\abs{\E[Z_i] - \E[Y_i]} \lsim \sigma_w^3\sigma_y^3 t^4 \exp(-t^2/4).
\end{align}
\end{lemma}
\begin{proof}
We upper bound $\abs{\E[Z_i] - \E[Y_i]}= \abs{\E[Y_i \ind{\Escr_i^c}]}$.
\begin{align}
6\abs{\E[Y_i \ind{\Escr_i^c}]} &\leq 6\Exp{\abs{Y_i} \ind{\Escr_i^c}} \notag\\
&\leq \Exp{\abs{y_i^3 w_i^3} \ind{\Escr_i^c}} + \Exp{\abs{y_i^3 w_i} \opnorm{W}^2 \ind{\Escr_i^c}}.\label{eq:intermed1}
\end{align}

Note that we can decompose the indicator of event $\Escr_i^c$ as $\ind{\Escr_i^c} = \ind{\Escr_{i,y}^c \cap \Escr_{i,w}^c} +
 \ind{\Escr_{i,y} \cap \Escr_{i,w}^c} + 
 \ind{\Escr_{i,y}^c \cap \Escr_{i,w}}$.
Focusing on the first term of \eqref{eq:intermed1}, we have 
\begin{align}
\Exp{\abs{y_i^3 w_i^3} \ind{\Escr_i^c}} &\leq\underbrace{\Exp{\abs{y_i^3 w_i^3} \ind{\Escr_{i,y}^c \cap \Escr_{i,w}^c}}}_{(A)}  + \underbrace{\Exp{\abs{y_i^3 w_i^3} \ind{\Escr_{i,y} \cap \Escr_{i,w}^c}}}_{(B)}  + \underbrace{\Exp{\abs{y_i^3 w_i^3} \ind{\Escr_{i,y}^c \cap \Escr_{i,w}}}}_{(C)}.
\end{align}

By the Cauchy-Schwarz inequality, 
\begin{align}
(A) &\leq \sqrt{ \Exp{y_i^6\ind{\Escr_{i,y}^c \cap \Escr_{i,w}^c}} \Exp{w_i^6\ind{\Escr_{i,y}^c \cap \Escr_{i,w}^c}} } \\
&\leq \sqrt{ \Exp{y_i^6\ind{\Escr_{i,y}^c}} \Exp{w_i^6 \ind{\Escr_{i,w}^c}}}. 
\end{align}
Corollary \ref{cor:subg-trunc-moment} implies
\begin{align}
\Exp{y_i^6 \ind{\abs{y_i}> t_y}} &\leq 12 \exp\lrp{-\frac{t_y^{2}}{2\sigma_y^2}} \lrp{8\sigma_y^6 + 4t_y^{2}\sigma_y^4 + t_y^{4} \sigma_y^2},
\end{align}
and similarly for $w_i$, so that 
\begin{align}
(A) &\leq 48 \exp\lrp{-\frac{t_y^{2}}{4\sigma_y^2} -\frac{t_w^{2}}{4\sigma_w^2}}
 \lrp{2\sigma_y^6 + t_y^{2}\sigma_y^4 + \frac{1}{4} t_y^{4} \sigma_y^2}^{1/2} \lrp{2\sigma_w^6 + t_w^{2}\sigma_w^4 + \frac{1}{4} t_w^{4} \sigma_w^2}^{1/2}\\
 &\lsim \sigma_y^3\sigma_w^3 t^4 \exp(-t^2/2),
\end{align}
where we have used $t = t_w/\sigma_w = t_y/\sigma_y \geq 1$ to simplify expressions.
Similarly, by the Cauchy-Schwarz inequality, a bound on the moment of the subgaussian variable $y_i$ (c.f. Proposition 2.5.2 of \cite{vershynin2018high} and from Corollary \ref{cor:subg-trunc-moment}, we have 
\begin{align}
(B) &= \Exp{\abs{y_i^3 w_i^3} \ind{\set{\abs{y_i} \leq t_y} \cap \set{\abs{w_i}> t_w}}}  \\
&\leq \sqrt{\Exp{y_i^6}  \Exp{w_i^6 \ind{\abs{w_i} > t_w}}} \\
&\leq  \lrp{6\sqrt{3}\sigma_y}^3 \cdot \lrp{48 \exp\lrp{-\frac{t_w^{2}}{2\sigma_w^2}} \lrp{2\sigma_w^6 + t_w^{2}\sigma_w^4 + \frac{1}{4} t_w^{4} \sigma_w^2}}^{1/2}\\
&\lsim \sigma_w^3 \sigma_y^3 t^2 \exp(-t^2/4).
\end{align}

Finally, 
\begin{align}
(C) &= \Exp{\abs{y_i^3 w_i^3} \ind{\abs{y_i}> t_y}\cap\ind{\abs{w_i} \leq t_w}}  \\
&\leq \sqrt{\Exp{y_i^6\ind{\abs{y_i}> t_y}}  \Exp{w_i^6}} \\
&\leq  \lrp{48 \exp\lrp{-\frac{t_y^{2}}{2\sigma_y^2}} \lrp{2\sigma_y^6 + t_y^{2}\sigma_y^4 + \frac{1}{4} t_y^{4} \sigma_y^2}}^{1/2} \cdot\lrp{6\sqrt{3}\sigma_w}^3 \\
&\lsim \sigma_w^3 \sigma_y^3 t^2 \exp(-t^2/4)
\end{align}
Altogether, we have 
\begin{align}
\Exp{\abs{y_i^3 w_i^3} \ind{\Escr_i^c}} &\lsim \sigma_w^3 \sigma_y^3 (t^4 + 2t^2) \exp(-t^2/4)
\end{align}

Following a similar procedure for the second term, we have 
\begin{align}
\Exp{\abs{y_i^3 w_i} \ind{\Escr_i^c}} &\leq \underbrace{\Exp{\abs{y_i^3 w_i} \ind{\Escr_{i,y}^c \cap \Escr_{i,w}^c}}}_{(A)}  + \underbrace{\Exp{\abs{y_i^3 w_i} \ind{\Escr_{i,y} \cap \Escr_{i,w}^c}}}_{(B)}  + \underbrace{\Exp{\abs{y_i^3 w_i} \ind{\Escr_{i,y}^c \cap \Escr_{i,w}}}}_{(C)}
\end{align}
with 
\begin{align}
(A) &\leq \sqrt{\Exp{y_i^6 \ind{\Escr_{i,y}^c}}\Exp{w_i^2 \ind{\Escr_{i,w}^c}}} \\
&\leq \sqrt{48 \exp\lrp{-\frac{t_y^{2}}{2\sigma_y^2}} \lrp{2\sigma_y^6 + t_y^{2}\sigma_y^4 + \frac{1}{4} t_y^{4} \sigma_y^2}}\sqrt{4\sigma_w^2 \exp \lrp{-\frac{t_w^2}{2\sigma_w^2}}} \\
&\lsim \sigma_w \sigma_y^3 t^2 \exp(-t^2/2), \\
(B) &\leq\sqrt{\Exp{y_i^6} \Exp{w_i^2 \ind{\Escr_{i,w}^c}}} \\
&\leq \sqrt{\lrp{6\sqrt{3}\sigma_y}^6}\sqrt{4\sigma_w^2 \exp \lrp{-\frac{t_w^2}{2\sigma_w^2}}}\\
&\lsim \sigma_w \sigma_y^3 \exp(-t^2/4), \\
(C) &\leq \sqrt{\Exp{y_i^6 \ind{\Escr_{i,y}^c}}\Exp{w_i^2}} \\
&\leq \sqrt{48 \exp\lrp{-\frac{t_y^{1/3}}{2\sigma_y^2}} \lrp{2\sigma_y^6 + t_y^{1/3}\sigma_y^4 + \frac{1}{4} t_y^{2/3} \sigma_y^2}} \sqrt{36\sigma_w^2} \\
&\lsim t^2\sigma_w \sigma_y^3 \exp(-t^2/2).
\end{align}

With $\opnorm{W}^2 = \sigma_w^2$, we combine results to get 
\begin{align}
\opnorm{W}^2\Exp{\abs{y_i^3 w_i} \ind{\Escr_i^c}} &\lsim \sigma_w^3 \sigma_y^3 (2t^2+1) \exp(-t^2/4)
\end{align}

In all (with $\opnorm{W}^2 = \sigma_w^2$), we have 
\begin{align}\label{eq:bigfat1}
\begin{split}
\abs{\E[Z_i] - \E[Y_i]} & \lsim \sigma_w^3 \sigma_y^3 (t^4+2t^2) \exp(-t^2/4) + \sigma_w^3 \sigma_y^3 (2t^2+1) \exp(-t^2/4) \\
&\lsim \sigma_w^3\sigma_y^3 t^4 \exp(-t^2/4).
\end{split}
\end{align}
\end{proof}

\subsection{Whitening Perturbation Bounds}
In Lemma \ref{lem:whiten-perturb}, Lemma \ref{lem:whiten-diff}, and Corollary \ref{cor:whitened}, we provide bounds to propagate error from matrices $M$ to their whitening matrices $W$, and to the third order tensors evaluated on different whitening matrices. 
\begin{lemma}[Lemma 9 in \cite{yi2016solving}]\label{lem:whiten-perturb}
Let $M$ and $\Mhat$ be positive semidefinite matrices in $\R^{d\by d}$, of rank $k$. Let $W, \What \in \R^{d\by K}$ be whitening matrices such that $WMW = I_K$ and $\What \Mhat \What = I_K$. Let $\alpha:= \opnorm{M - \Mhat}/\sigma_k(M)$. When $\alpha<1/3$, we have that
\begin{align}
\frac{1}{3}\opnorm{W} &\leq \opnorm{\What} \leq 2\opnorm{W} \\
\opnorm{W - \What} &\leq 2\alpha \opnorm{W} \\
\opnorm{\pseudo{\What}} &\leq 2\opnorm{\pseudo{W}} \\
\opnorm{\pseudo{W} - \pseudo{\What}} &\leq 2\alpha \opnorm{\pseudo{W}}. 
\end{align}
\end{lemma}

\begin{lemma}\label{lem:whiten-diff}
Let $\Mhat$ be a $d\by d \by d$ symmetric tensor, and let $W$ and $\What$ be $d \by K$ matrices. Then 
\begin{align}
&\opnorm{\Mhat(\What, \What, \What) - \Mhat(W, W, W)} \\
\leq &5\lrp{\opnorm{\What}^2 + \opnorm{\What}\opnorm{W} + \opnorm{W}^2} \opnorm{\What - W} \opnorm{\Mhat}
\end{align}
\end{lemma}
\begin{proof}
Beginning with the definition of operator norm for a symmetric tensor, we have
\begin{align}
\opnorm{\Mhat(\What, \What, \What) - \Mhat(W, W, W)}&\leq \sup_{v \in \Sscr^{K-1}} \abs{\Mhat(\What v, \What v, \What v)-M(Wv, Wv, Wv)} 
\end{align}
For any $v \in \Sscr^{K-1}$, we have
\begin{align}
&\abs{\Mhat(\What v, \What v, \What v)-M(Wv, Wv, Wv)} \\
\leq & \Mhat((\What-W)v, \What v, \What v) + \Mhat(W v, (\What-W)v, \What v)+ \Mhat(W v, W v, (\What-W)v) \\
\leq &\opnorm{\What-W}\lrp{\opnorm{\What}^2 + \opnorm{\What}\opnorm{W} + \opnorm{W}^2} \lrp{5\opnorm{\Mhat}}
\end{align}
where in the last line we normalized the arguments of $\Mhat$ to be unit vectors and used Lemma \ref{lem:tensor_norm} to bound the expressions by $\opnorm{\Mhat}$.
\end{proof}

The following result shows how a perturbation of $M_2$ propagates through to a perturbation of $M_3^W$ through a perturbation of the whitening matrix $W$. In particular, the conditions for Corollary \ref{cor:whitened} hold with probability at least $1-\delta$ when $N_2$ satisfies \eqref{eq:N2-condition} with $\eps = \sigma_K(M_2)/3$.
\begin{corollary}\label{cor:whitened}
Let $W$ and $\What$ be whitening $d\by K$ matrices for the $d\by d$ matrices $M_2$ and $\Mhat_2$, respectively. 
When $\opnorm{M_2 - \Mhat_2} \leq \sigma_K(M_2)/3$, we have that for any third order symmetric tensor $M_3 \in \R^{d\by d \by d}$,
\begin{align}
\opnorm{\mixed - \clean} &\lsim  \frac{\opnorm{\Mhat_2-M_2}}{\sigma_K(M_2)^{5/2}} \opnorm{M_3}.
\end{align}
\end{corollary}

\begin{proof}
Recall that for ease of notation we set $\sigma_K:= \sigma_K(M_2)$. 
Under this condition that $\opnorm{M_2 - \Mhat_2} \leq \sigma_K/3$, from Lemma \ref{lem:whiten-perturb} and recalling that $\opnorm{W}^2 = \sigma_K$, we have that
\begin{align}
\opnorm{W - \What} \leq 2 \frac{\opnorm{M_2 - \Mhat_2}}{\sigma_K^{3/2}}, 
\end{align}
and that $\opnorm{\What} \leq 2\opnorm{W} = 2\sigma_K^{-1/2}$.

Next, from Lemma \ref{lem:whiten-diff},
\begin{align}
\opnorm{\mixed - \clean} &\leq 5\lrp{\opnorm{\What}^2 + \opnorm{\What}\opnorm{W} + \opnorm{W}^2} \opnorm{\What - W} \opnorm{M_3} \\
&\leq 70 \opnorm{W}^2 \frac{\opnorm{M_2 - \Mhat_2}}{\sigma_K^{3/2}} \opnorm{M_3} \\
&\lsim \frac{\opnorm{M_2 - \Mhat_2}}{\sigma_K^{5/2}} \opnorm{M_3}
\end{align}
\end{proof}

\subsection{Tensor Decomposition Algorithm and Lemmas}
In this section, we provide the tensor power iteration method for tensor decomposition from \cite{anandkumar2014tensor} in Algorithm \ref{alg:tensor-power-method}, along with results on the robustness of the method. We use the presentation of Algorithm 2 and Lemma 4 in \cite{yi2016solving}, which are restatements of Algorithm 1 and Theorem 5.1 of \cite{anandkumar2014tensor}.
\IncMargin{1em}%
\begin{algorithm2e}[ht]
\caption{Robust Tensor Power Method \citep[Algorithm 1]{anandkumar2014tensor}}\label{alg:tensor-power-method}
\SetAlgoLined
\LinesNumbered
\DontPrintSemicolon
\KwIn{Symmetric tensor $M \in \R^{K \by K \by K}$, $K$ \\
Parameters $\Rstart$ - number of starting points, and $\Riter$ - number of iterations.}
\KwOut{$\set{(\phat_j, \betahat_j)\mid  j \in [K]}$ such that $M \approx \sum_{j=1}^K \phat_j \betahat_j^{\ot 3}$ and $\enorm{\betahat_j} = 1$ for $j \in [K]$.}

\For{$j = 1, \dots, K$}{
    \For(\tcp*[f]{Iterate on $\Rstart$ initial points}){$l = 1, \dots, \Rstart$}{ 
        $\beta_0^{(l)} \sim \Unif(\Sscr^{K-1})$ \;
        \For(\tcp*[f]{$\Riter$ power iterations}){$t = 0, \dots, \Riter$}{
            $\beta_{t+1}^{(l)} \gets M(I_K, \beta_t^{(l)}, \beta_t^{(l)}) = \sum_{i=1}^d \sum_{j, k \in [d]} M_{i,j,k} u_j u_k \vec{e}_i $ \;
            $\beta_{t+1}^{(l)} \gets \beta_{t+1}^{(l)}/\enorm{\beta_{t+1}^{(l)}}$ \;
        }
    }
    $l^* \gets \amax_{l \in [L]} M(\beta_\Riter^{(l)}, \beta_\Riter^{(l)}, \beta_\Riter^{(l)})$ \;
    $\beta_0 \gets \beta_\Riter^{(l^*)}$ \tcp*[r]{$\Riter$ more power updates on the best point}
    \For{$t = 0, \dots, \Riter$}{
        $\beta_{t+1}^{(l)} \gets M(I_K, \beta_t^{(l)}, \beta_t^{(l)}) = \sum_{i=1}^d \sum_{j, k \in [d]} M_{i,j,k} u_j u_k \vec{e}_i $ \;
        $\beta_{t+1}^{(l)} \gets \beta_{t+1}^{(l)}/\enorm{\beta_{t+1}^{(l)}}$ \;
    }
    $\betahat_j \gets \beta_\Riter^{(l^*)}$ \;
    $\phat_j \gets M(\betahat_j, \betahat_j, \betahat_j)$ \;
    $M \gets M - \phat_j \betahat_j^{\ot 3}$ \;
}
\end{algorithm2e}
\DecMargin{1em} %

\begin{lemma}[Robust Tensor Power Method, Theorem 5.1 in \cite{anandkumar2014tensor}] \label{lem:tensor-robustness}
Suppose $M \in \R^{K \times K \times K}$ is a tensor with decomposition $M = \sum_{k=1}^K p_k \beta_k^{\ot 3}$ where $\set{\beta_k}$ are orthonormal. Let $\pmin := \min_{k \in [K]} \set{p_k}>0$. Let $\widehat{M} = M + E$ be the input of Algorithm \ref{alg:tensor-power-method}, where $E$ is a symmetric tensor with $\opnorm{E}\leq \eps$. There exist constants $C_1, C_2, C_3>0$ such that the following holds. Suppose $\eps \leq C_1 \pmin / K$. For any $\delta \in (0, 1)$, suppose $(\Riter, \Rstart)$ in Algorithm \ref{alg:tensor-power-method} satisfies
\begin{align}\label{eq:alg-param-values}
\Riter \geq C_2 \cdot \lrp{\log K + \log \log (1/ \epsilon)}, \quad \Rstart \geq C_3 \cdot \poly(K) \log (1 / \delta),
\end{align}
for some polynomial function $\mathrm{poly}(\cdot)$. 
With probability at least $1 - \delta$, $\set{\phat_j, \widehat{\beta}_j)}$ returned by Algorithm \ref{alg:tensor-power-method} satisfy the bound
\begin{align}
\enorm{\widehat{\beta}_j - \beta_{\pi(j)}} \leq \frac{8 \epsilon}{p_{\pi(j)}}, \quad \abs{\phat_j - p_{\pi(j)}} \leq 5 \eps \xpln{for all $j \in [K]$},
\end{align}
where $\pi(\cdot)$ is some permutation function on $[K]$.
\end{lemma}

\subsection{Tensor Norm Lemmas}
In Lemmas \ref{lem:tensor-factor}, \ref{lem:tensor_norm}, and \ref{lem:covering}, we provide expressions for bounding the norms of third order tensors.  
\begin{lemma}\label{lem:tensor-factor}
Let $T$ be a symmetric $d\by d\by d$ tensor, and let $W$ be a $d\by K$ matrix. Then $\opnorm{T(W,W,W)}\leq \opnorm{W}^3 \opnorm{T}$. 
\end{lemma}
\begin{proof}
Starting with the definition, 
\begin{align}
\opnorm{T}(W, W, W) &= \sup_{v \in \Sscr^{K-1}} \abs{T(Wv, Wv, Wv)} \\
&\leq \sup_{v \in \Sscr^{K-1}} \enorm{Wv}^3 \abs{T(u, u, u)} \xpln{$u:= Wv/\enorm{Wv}$} \\
&\leq \opnorm{W}^3 \sup_{u \in \Sscr^{d-1}} \abs{T(u,u,u)} = \opnorm{W}^3 \opnorm{T}.
\end{align}
\end{proof}

The following lemma is based on Lemma 12 in \cite{yi2016solving} but with an improved constant. 
\begin{lemma}[Tensor operator norm]\label{lem:tensor_norm}
For any symmetric third-order tensor $T \in \R^{d \by d \by d}$, 
\begin{align}
\opnorm{T} \leq \sup_{a,b,c \in \Sscr^{d-1}} T(a,b,c) \leq 5 \opnorm{T}
\end{align}
\end{lemma} 
\begin{proof}
Note that for all $a, b, c \in \Sscr^{d-1}$, 
\begin{align}
T(a+b, a+b, c) &= T(a, a, c) + T(a, b, c) + T(b, a, c) + T(b, b, c).
\end{align}
Rearranging terms and using the symmetry of $T$ and that $\enorm{a+b} \leq 2$, we have that 
\begin{align}
2 T(a, b, c) \leq \sup_{u, v\in \Sscr^{d-1}} (2^2 + 1 + 1) T(u, u, v) = \sup_{u,v\in \Sscr^{d-1}} 6T(u,u,v).
\end{align}
Next, 
\begin{align}
T(u+v, u+v, u+v) &= T(u,u,u) + T(v,v,v) + 3T(u,u,v) - 3T(u,v,v),
\end{align}
which implies that 
\begin{align}
6\sup_{u,v\in \Sscr^{d-1}} T(u,u,v) \leq (2^3 + 1 + 1)\sup_{u\in \Sscr^{d-1}} T(u,u,u).
\end{align}
In all we have that 
\begin{align}
T(a, b, c) \leq 3T(u,u,v) \leq 5 \sup_{u\in \Sscr^{d-1}} T(u,u,u).
\end{align}
\end{proof}

\begin{lemma}[Covering Lemma] \label{lem:covering}
Let $T$ be a symmetric $d \by d \by d$ tensor, $\eps \in (0, 1/2)$, and $\Cscr$ be an $\eps$-cover of $\Sscr^{d-1}$. Then 
\begin{align}
\sup_{v \in \Cscr} T(v,v,v) \leq \opnorm{T} \leq \frac{1}{1-15\eps} \sup_{v \in \Cscr} T(v, v, v).
\end{align}
\end{lemma}
\begin{proof}
Since $\Sscr^{d-1}$ is compact and the map $v \mapsto T(v, v, v)$ is continuous, there exists a $v^* \in \Sscr^{d-1}$ such that $\opnorm{T} = T(v^*, v^*, v^*)$. Let $v_0 \in \Cscr$ be such that $\enorm{v_0 - v^*} \leq \eps$. Then 
\begin{align}
T(v_0, v_0, v_0) - T(v*, v*, v*) &=T(v_0-v*, v_0, v_0) + T(v*, v_0-v*, v_0) + T(v*, v*, v_0-v*) \\
&= \enorm{v_0 - v*} \lrp{T(\delta, v_0, v_0) + T(v*, \delta, v_0) + T(v*, v*, \delta)}
\end{align}
where $\delta:= (v_0 - v*)/\enorm{v_0 - v*} \in \Sscr^{d-1}$. 
Since for every $a, b, c \in \Sscr^{d-1}$, $T(a, b, c) \leq 5 \opnorm{T} = T(v^*, v^*, v^*)$ by Lemma \ref{lem:tensor_norm}, we have 
\begin{align}
\abs{T(v_0, v_0, v_0) -\opnorm{T}} \leq 15\eps \opnorm{T}.
\end{align}
Rearranging terms gives us our claim.
\end{proof}
\subsection{Expectation of Truncated Subgaussian RVs}
In Corollary \ref{cor:subg-trunc-moment}, we provide bounds on the expectation of the truncated upper tails of subgaussian random variables and their powers. The result relies on Lemma \ref{lem:trunc-gauss-moment}, Lemma \ref{lem:intermed-subg-trunc} and Corollary \ref{cor:intermed-subg-trunc}, which we state first.

Lemma \ref{lem:trunc-gauss-moment} is similar to Lemma 14 in \cite{yi2016solving} but we extend the result to cover odd values of $p$ as well.
\begin{lemma}[Recursive Truncated Gaussian Moments]\label{lem:trunc-gauss-moment}
Let $X\sim \Nscr(0,1)$, and let $M_p(\tau):= \Exp{X^p \ind{X>\tau}} $ for all $\tau \geq 0$. 
\begin{align}
M_0(\tau) &= \Pr{X>\tau} \leq \frac{1}{\sqrt{2\pi}}\frac{1}{\tau} \exp\lrp{-\tau^2/2},\\
M_1(\tau) &= \sqrt{\frac{1}{2\pi}} e^{-\tau^2/2} \xpln{and}\\
M_p(\tau) &= (p-1) M_{p-2}(\tau) + \sqrt{\frac{1}{2\pi}} \tau^{p-1} \exp \lrp{-\tau^2/2} \xpln{for $p\geq 2$.}
\end{align}
\end{lemma}
\begin{proof}
Using integration by parts, let $v = x^{p+1}/(p+1)$ and $u = \frac{1}{\sqrt{2\pi}} \exp(-x^2/2)$. Then 
\begin{align}
M_p(\tau) &= \int_{x>\tau} \frac{1}{\sqrt{2\pi}} x^p \exp^{-x^2/2}\, dx =\int_\tau^\infty u\, dv \\
&= 2\lrp{ uv|_{\tau}^{\infty} -\int_\tau^\infty v\, du} \\
&= \lrp{\frac{1}{\sqrt{2\pi}} \frac{\tau^{p+1}}{p+1} \exp(-\tau^2/2)) + \frac{1}{p+1} \int_\tau^\infty \frac{1}{\sqrt{2\pi}} x^{p+2} \exp^{-x^2/2}dx} \\
&= \sqrt{\frac{1}{2\pi}} \frac{\tau^{p+1}}{p+1} \exp(-\tau^2/2)) + \frac{1}{p+1} M_{p+2}(\tau).
\end{align}
Given $M_0(\tau)$, the above gives us an expression for $M_p(\tau)$ for all even $p$. A bound on $M_0(\tau)$ can be obtained from Mill's inequality. Likewise, given $M_1(\tau)$ we obtain expressions for $M_p(\tau)$ for all odd $p$. We solve for $M_1(\tau)$ directly: 
\begin{align}
M_1(\tau) &= \Exp{X\ind{X>\tau}} \\
&= \int_{x>\tau} \frac{1}{\sqrt{2\pi}} \exp(-x^2/2) x\, dx \\
&= \int_{\tau^2/2}^\infty \frac{1}{\sqrt{2\pi}} \exp(-u)\, du \xpln{change variable $u = x^2/2$}\\
&= \sqrt{\frac{1}{2\pi}} e^{-\tau^2/2}. 
\end{align}
\end{proof}

Lemma \ref{lem:intermed-subg-trunc} can be considered as a corollary of Lemma 14 in \cite{yi2016solving}, though our proof is self-contained. 
\begin{lemma}\label{lem:intermed-subg-trunc}
Let $X$ be a subgaussian random variable with variance proxy $\sigma^2$. Then for every $a \geq 1$ and $\tau>0$, 
\begin{align}
\Exp{\abs{X}^a \1(\abs{X}> \tau)} \leq 2\sqrt{2\pi} \cdot a \cdot \sigma^a \cdot M_{a-1}\lrp{\frac{\tau}{\sigma}}
\end{align} 
where 
\begin{align}
M_0(\tau) &\leq \frac{1}{\sqrt{2\pi}}\frac{1}{\tau} \exp\lrp{-\tau^2/2},\\
M_1(\tau) &= \frac{1}{\sqrt{2\pi}} \exp(-\tau^2/2) \xpln{and}\\
M_p(\tau) &= (p-1) M_{p-2}(\tau) + \frac{1}{\sqrt{2\pi}} \tau^{p-1} \exp \lrp{-\tau^2/2} \xpln{for $p\geq 2$.}
\end{align}
\end{lemma}
\begin{proof}
\begin{align}
\Exp{\abs{X}^a \ind{\abs{X}>\tau}} &= \int_{\tau^a}^{\infty} \Pr{\abs{X}^a > t} dt \\
&= \int_{\tau^a}^{\infty} \Pr{\abs{X} > t^{1/a}} dt \\
&\leq \int_{\tau^a}^\infty 2\exp\lrp{\frac{-t^{2/a}}{2\sigma^2}} dt \xpln{$X$ subgaussian} \\
&= \int_{\tau}^\infty 2 a u^{a-1} \exp\lrp{-\frac{u^2}{2\sigma^2}} du \xpln{change variables $u = t^{1/a}$} \\
&= 2a\sqrt{2\pi\sigma^2} \Exp{U^{a-1}\ind{U>\tau}} \xpln{where $U \sim \Nscr(0,\sigma^2)$} 
\end{align}

Define $M_p(\tau):= \Exp{Z^p \ind{Z>\tau}}$ where $Z \sim \Nscr(0,1)$, for $p \geq 1$ and for all $\tau\geq 0$. Lemma \ref{lem:trunc-gauss-moment} gives upper bounds or exact values for $M_p(\tau)$, for $p$ even and $p$ odd, respectively. 
We then have
\begin{align}
\Exp{\abs{X}^a \ind{\abs{X}>\tau}} \leq  2\sqrt{2\pi}a\sigma^a M_{a-1}\lrp{\frac{\tau}{\sigma}}.
\end{align}
\end{proof}

\begin{corollary}\label{cor:intermed-subg-trunc}
\begin{align}
M_2(\tau) &\leq \sqrt{\frac{1}{2\pi}}\exp(-\tau^2/2)\lrp{\frac{1}{\tau}+\tau} \\
M_3(\tau) &= \sqrt{\frac{1}{2\pi}} \exp(-\tau^2/2) \lrp{2+\tau^2} \\
M_4(\tau) &\leq  \sqrt{\frac{1}{2\pi}} \exp(-\tau^2/2) \lrp{\frac{3}{\tau} + 3\tau + \tau^3}\\
M_5(\tau) &=  \sqrt{\frac{1}{2\pi}} \exp(-\tau^2/2) \lrp{8 + 4\tau^2 + \tau^4}\\
M_6(\tau) &\leq  \sqrt{\frac{1}{2\pi}} \exp(-\tau^2/2) \lrp{\frac{15}{\tau} + 15\tau + 5\tau^3 + \tau^5}
\end{align}
\end{corollary}

\begin{corollary}\label{cor:subg-trunc-moment}
Let $X$ be a subgaussian random variable with variance proxy $\sigma^2$. Then: 
\begin{align}
\Exp{\abs{X}^1 \1(\abs{X}> \tau)} &\leq \frac{2\sigma^2}{\tau} \exp\lrp{-\frac{\tau^2}{2\sigma^2}} \\
\Exp{\abs{X}^2 \1(\abs{X}>\tau)} &\leq 4\sigma^2 \exp\lrp{-\frac{\tau^2}{2\sigma^2}} \\
\Exp{\abs{X}^3 \1(\abs{X}>\tau)} &\leq 6\lrp{\frac{\sigma^4}{\tau} + \sigma^2 \tau}\exp\lrp{-\frac{\tau^2}{2\sigma^2}}  \\
\Exp{\abs{X}^4 \1(\abs{X}>\tau)} &\leq 8 \lrp{2\sigma^4 + \sigma^2\tau^2}\exp\lrp{-\frac{\tau^2}{2\sigma^2}}  \\
\Exp{\abs{X}^6 \1(\abs{X}>\tau)} &\leq 12  \lrp{8\sigma^6 + 4\tau^2\sigma^4 + \tau^4\sigma^2} \exp\lrp{-\frac{\tau^2}{2\sigma^2}} 
\end{align}
\end{corollary}

\subsection{Additional notes on implementation}\label{sec:simulation-details}
Simulations were implemented in \texttt{Matlab R2020b}. Some tensor operations were implemented using \texttt{tensor\_toolbox} \citep{tensortoolbox}. Empirically, tensor power iteration had better accuracy than other tensor decomposition methods such as the alternating least squares and orthogonalized alternating least squares estimators provided by \texttt{tensor\_toolbox}.

Because of the sensitivity of the tensor decomposition approach to the conditioning of the mixture parameters and regression covariates, as well as the high order sample complexity for estimating the moments in $M_2$ and $M_3$, the recovered mixture weights and Markov parameters after the dewhitening step can still be very noisy. We propose refining these estimates by solving a constrained linear equation with the estimated first order moment $\Mhat_1 = \sum_{i,t} u_{i,t} y_{i,t}$, which has better concentration properties than higher order moments, by finding the weights and parameters $\set{\phat_k, \beta_k}$ such that $\sum_{k=1}^{K} \phat_k \beta_k = \Mhat_1$ with $\sum_k \phat_k = 1$. Empirically, this slightly improves our estimates. It would be interesting future work to prove formal results on post-processing methods to improve the robustness of tensor decomposition to the conditioning of the mixture regression parameters. 
\end{document}